\documentclass[12pt]{article}
\usepackage[utf8]{inputenc}
\usepackage{amsfonts,latexsym,amsmath,amsthm,amssymb,amscd,euscript,mathrsfs,graphicx}
\usepackage{framed}
\usepackage{fullpage}
\usepackage{color}
\usepackage{enumitem}
\usepackage[obeyFinal,textsize=scriptsize,shadow]{todonotes}
\usepackage{tikz}
\usetikzlibrary{matrix}
\usepackage[hidelinks]{hyperref}
\usepackage{cleveref}
\usepackage{soul}

\theoremstyle{definition}
\newtheorem{theorem}{Theorem}[section]
\newtheorem{proposition}[theorem]{Proposition}
\newtheorem{lemma}[theorem]{Lemma}
\newtheorem{corollary}[theorem]{Corollary}

\newtheorem{defn}[theorem]{Definition}

\newtheorem{question}[theorem]{Question}
\newtheorem{problem}[theorem]{Problem}

\theoremstyle{remark}

\newcommand{\BC}{\mathbb C}
\newcommand{\BE}{\mathbb E}

\newcommand{\BI}{\mathbb I}

\newcommand{\BP}{\mathbb P}

\newcommand{\BR}{\mathbb R}
\newcommand{\BZ}{\mathbb Z}
\newcommand{\eps}{\varepsilon}
\newcommand{\poly}{\text{poly}}
\newcommand{\cD}{\mathcal{D}}
\DeclareMathOperator{\Trunc}{Trunc}
\DeclareMathOperator{\tr}{Tr}

\DeclareMathOperator{\Err}{Err}
\DeclareMathOperator{\supp}{supp}

\newcommand{\samples}{\mathfrak{n}}
\newcommand{\locality}{\mathfrak{K}}
\newcommand{\dualgraph}{\mathfrak{G}}
\newcommand{\degree}{\mathfrak{d}}
\newcommand{\bit}{B}

\title{Improved algorithms for learning quantum Hamiltonians, via flat polynomials}
\author{Shyam Narayanan\thanks{\texttt{shyam.s.narayanan@gmail.com}. Massachusetts Institute of Technology. Supported by an NSF Graduate Fellowship and a Google Fellowship.}}
\date{\today}

\begin{document}

\maketitle

\begin{abstract}
    We give an improved algorithm for learning a quantum Hamiltonian given copies of its Gibbs state, that can succeed at any temperature. Specifically, we improve over the work of Bakshi, Liu, Moitra, and Tang~\cite{bakshi2024quantum}, by reducing the sample complexity and runtime dependence to singly exponential in the inverse-temperature parameter, as opposed to doubly exponential. Our main technical contribution is a new flat polynomial approximation to the exponential function, with significantly lower degree than the flat polynomial approximation used in~\cite{bakshi2024quantum}.
\end{abstract}

\section{Introduction}

Hamiltonian learning is an important problem at the intersection of quantum algorithms and quantum machine learning. The goal of Hamiltonian learning is to estimate the Hamiltonian of a quantum system given multiple independent copies of its Gibbs state.
This problem has been well-studied, with many theoretical and experimental works (e.g.,~\cite{wiebe2014hamiltonian, wang2017experimental,bairey2019learning,evans2019scalable,anshu2020sample, bakshi2024quantum}). This problem also has applications to areas including superconductivity and condensed matter physics (see~\cite{bakshi2024quantum} for further discussion).

The Hamiltonian learning problem that we study can be roughly viewed as follows. Suppose we have $n$ interacting qubits. The goal is to (approximately) learn the Hamiltonian $H = \sum \lambda_a E_a \in \BC^{2^n \times 2^n}$, where we assume each interaction term $E_a$ is known and each $\lambda_a$ represents the strength of the corresponding interaction (see Definitions~\ref{def:local-term} and~\ref{def:Hamiltonian}). We assume we can sample multiple i.i.d. copies from the Gibbs state $\rho \propto e^{-\beta \cdot H}$, where $\beta$ is the inverse temperature. The goal is to, using as few copies and/or as little time as possible, approximately learn the Hamiltonian matrix $H$, which amounts to estimating each parameter $\lambda_a$ if the interaction terms $E_a$ are known.

For general interactions of qubits, this problem can be intractable. To combat this, we assume (as in~\cite{haah2022optimal, bakshi2024quantum}) that $H$ is what is called a \emph{low-interaction} Hamiltonian (see \Cref{def:low-interaction}). This can, for instance, capture the interactions of qubits on a lattice of any constant number of dimensions.

\paragraph{Prior work:}

There has been substantial previous work towards the question of learning $H$ from copies of the Gibbs state. 

Consider a system of $n$ qubits with $m$ interaction terms $E_a$ and interaction strengths $\lambda_a$, at inverse temperature $\beta$.
Anshu, Arunachalam, Kuwahara, and Soleimanifar~\cite{anshu2020sample} proved that one can learn each term $\lambda_a$ up to error $\eps$, using only $\frac{e^{\poly(\beta)} \cdot m^3 \log m}{\beta^{O(1)} \eps^2}$ copies, but with a computationally inefficient algorithm. A follow-up work by the same authors~\cite{anshu2021efficient} gave an efficient algorithm, but only when the interaction terms all commute with each other. Returning to the general case (i.e., the interaction terms do not necessarily commute), a subsequent work by Haah, Kothari, and Tang~\cite{haah2022optimal} proved that, if $\beta < \beta_c$ for some critical threshold $\beta_c$, there is an algorithm that works with $\frac{\log m}{\beta^2 \eps^2}$ copies and $\frac{m \log m}{\beta^2 \eps^2}$ time.

However, the question of whether a polynomial-time algorithm for learning quantum Hamiltonians at lower temperature (i.e., $\beta > \beta_c$) was still open. Bakshi, Liu, Moitra, and Tang~\cite{bakshi2024quantum} resolved this question, by proving that for any fixed constant $\beta_c > 0$, there is an algorithm that requires $\poly\left(m, (1/\eps)^{e^{O(\beta)}}\right)$ samples and time for all $\beta \ge \beta_c$.

While this implies a polynomial-time algorithm for any fixed temperature $\beta$, the doubly exponential dependence on $\beta$ is somewhat unfortunate. While an $e^{O(\beta)}$ dependence is known to be necessary~\cite{haah2022optimal}, there is no inherent reason that an algorithm at low temperature must require a doubly-exponential dependence on $\beta$. Indeed, the authors of~\cite{bakshi2024quantum} ask the following as their main open question.

\begin{question}
    Is it possible to achieve a polynomial-time algorithm for learning quantum Hamiltonians with runtime that is only singly exponential in $\beta$?
\end{question}

\paragraph{Our work.}
In this work, we resolve this problem by proving that there exists an algorithm that only requires $\poly(m, (1/\eps)^{O(\beta^2)})$ samples and time. Hence, we have reduced the dependence to only \emph{singly} exponential, rather than \emph{doubly} exponential, in $\beta$.

The high-level outline of the algorithm is in fact the same as~\cite{bakshi2024quantum}. The main bottleneck, however, in the work of~\cite{bakshi2024quantum}, was that their algorithm's runtime crucially depends exponentially on the degree of a certain ``flat exponential approximation'' polynomial. The degree of the polynomial they constructed was $e^{O(\beta)} \cdot \log \frac{1}{\eps}$. Our main contribution is to construct a novel polynomial which satisfies the same flat exponential approximation guarantees (as well as some additional guarantees needed by~\cite{bakshi2024quantum}), while having degree only $O(\beta^2 \cdot \log \frac{1}{\eps})$.

\subsection{Problem Statement}

In this subsection, we formally define the problem that we are studying. We will formally state the main theorem in the next subsection.
Before we define the problem, we first provide some background, by defining local terms, Hamiltonians, and low-intersection Hamiltonians.

First, we will set $N = 2^n,$ and consider the space $\BC^N = \underbrace{\BC^2 \otimes \BC^2 \otimes \cdots \otimes \BC^2}_{n \text{ times}}$.

\begin{defn}[Local Term] \label{def:local-term}
    Fix a subset $S \subset [n]$, and consider a Hermitian matrix $E \in \BC^{2^{|S|} \times 2^{|S|}}$. We can view $E$ as a matrix in $\BC^{N \times N}$, by taking the tensor product of $E$ with the $2 \times 2$ identity matrix $I_2$ for all coordinates $j \in [n] \backslash S$.

    We will call such a matrix $E \in \BC^{N \times N}$ a \emph{local term}, and we define $\supp(E)$ to be the corresponding set $S$.
\end{defn}

\begin{defn}[Hamiltonian] \label{def:Hamiltonian}
    A \emph{Hamiltonian} is a Hermitian matrix $H \in \BC^{N \times N}$ that can be written as linear combination of local terms $E_a$ with associated coefficients $\lambda_a$, where $a$ ranges from $1$ to $m$. In other words, $H = \sum_{a=1}^m \lambda_a E_a$. For normalization purposes, we assume that $|\lambda_a| \le 1$ and $\|E_a\|_{op} \le 1$ for all $a \in [m]$.
    
    Finally, we say that $H = \sum_{a=1}^m \lambda_a E_a$ is $\locality$-\emph{local} if every term $E_a$ satisfies $\supp(E_a) \le \locality$.
\end{defn}

\begin{defn}[Low-interaction Hamiltonian~\cite{haah2022optimal}] \label{def:low-interaction}
    For a Hamiltonian $H = \sum_{a=1}^m \lambda_a E_a$ on a system of $n$ qubits, its \emph{dual interaction} graph $\dualgraph(H)$ is an undirected graph on $[m]$, with an edge between $a \neq b \in [m]$ if and only if $\supp(E_a) \cap \supp(E_b) \neq \emptyset$. 
    
    We say that $H = \sum_{a=1}^m \lambda_a E_a$ is a \emph{low-interaction Hamiltonian} if every $E_a$ is $\locality$-local and if the maximum degree of the graph $\dualgraph(H)$ is some $\degree$, and $\locality, \degree$ are bounded by a fixed constant.
\end{defn}

Finally, we are ready to define the formulation of Hamiltonian learning that we study.

\begin{problem} \label{prob:main}
    Let $H = \sum_{a=1}^m \lambda_a E_a$ be a Hamiltonian on a system of $n$ qubits, where the local terms $E_a$ are known but the coefficients $\lambda_a \in [-1, 1]$ are unknown. Let $\eps, \beta > 0$ be some known parameters, corresponding to accuracy and inverse temperature, respectively.
    
    Now, given $\samples$ copies of the Gibbs state $\rho = \frac{\exp(-\beta H)}{\tr(\exp(-\beta H))}$, the goal of \emph{Hamiltonian learning} is to provide estimates $\hat{\lambda}_a$ such that, with probability at least $2/3$, $|\hat{\lambda}_a - \lambda_a| \le \eps$ for all $a \in [m]$.

    The goal is to solve this problem while minimizing $\samples$ and the runtime of the algorithm providing the estimates.
\end{problem}

We remark that the $2/3$ probability is arbitrary: we can improve it to probability $1-\delta$ by a simple amplification trick, needing only $\log (1/\delta)$ times as many samples and as much time.

\subsection{Main Theorem}

We can now formally state our main theorem, which improves the best-known results for \Cref{prob:main}.

\begin{theorem} \label{thm:main}
    Let $H = \sum_{a=1}^m \lambda_a E_a \in \BC^{N \times N}$ be a low-interaction Hamiltonian on $n$ qubits (i.e., the locality $\locality$ and maximum degree $\degree$ of $\dualgraph(H)$ are bounded by some fixed constant). Given knowledge of $E_a$, $\eps \in (0, 1)$, and $\beta > 0$, there is an algorithm that can output estimates $\hat{\lambda}_a$, such that $|\hat{\lambda}_a - \lambda_a| \le \eps$ for all $a \in [m]$, with at least $2/3$ probability. Moreover, the algorithm uses
\[\samples = O\left(m^6 \cdot (1/\eps)^{O(\beta^2)} + \frac{\log m}{\beta^2 \eps^2}\right)\]
    copies of the Gibbs state and runtime
\[O\left(m^{O(1)} \cdot (1/\eps)^{O(\beta^2)} + \frac{m \log m}{\beta^2 \eps^2}\right).\]
    The big $O$ notation may hide dependencies on $\locality$ and $\degree$, which are assumed to be constant.
\end{theorem}

This improves over the previous work of~\cite{bakshi2024quantum}, which had the $(1/\eps)^{O(\beta^2)}$ terms (in both the sample complexity and runtime) replaced with $(1/\eps)^{e^{O(\beta)}}$.

The proof of \Cref{thm:main} is based on a new ``flat exponential approximation,'' which is the main technical contribution of this work, and is a result that we believe may be of independent interest. We state the flat exponential approximation result here, though we remark that we prove a more general result in~\Cref{thm:approx-main} (which also deals with additional constraints needed for the quantum Hamiltonian learning problem).

\begin{theorem} \label{thm:approx-simplified}
    Let $\eps_0 > 0$ be a sufficiently small constant. Then, for any $0 < \eps < \eps_0$ and any $\beta \ge 1$, there exists a polynomial $P(x)$ of degree at most $O(\beta^2 \log \frac{1}{\eps})$, such that
\begin{enumerate}
    \item For all $x \in [-\beta \log \frac{1}{\eps}, \beta \log \frac{1}{\eps}]$, $|P(x)-e^{-x}| \le \eps$.
    \item For all $x \le 0$, $|P(x)| \le e^{|x|}$.
    \item For all $x \ge 0$, $|P(x)| \le e^{|x|/\beta}$.
\end{enumerate}
\end{theorem}

Intuitively, this polynomial serves as a very good approximation to the exponential function $e^{-x}$ in a decently sized interval around $0$. Moreover, this function does not grow faster than the exponential on the negative side (where $e^{-x}$ blows up), but does not grow faster than a slow exponential rate on the positive side (where $e^{-x}$ decays).

\subsection{Technical Overview}

We discuss the main ideas in improving over the work of Bakshi et al.~\cite{bakshi2024quantum}.

\subsubsection{Flat exponential approximation.} 

First, we discuss how to obtain \Cref{thm:approx-simplified}, which is the main technical ingredient in our improvement. Before discussing our improvement, we first discuss the high-level approach of \cite{bakshi2024quantum}, which was inspired by the methods of ``peeling'' the exponential used in the proof of the classic Lieb-Robinson bound~\cite{liebrobinson, hastings2010locality}.

For simplicity, we focus on the following slightly weaker goal. We wish to construct a polynomial $P$ of low degree, such that $|P(x)-e^{-x}| \le \eps$ for all $x \in [-1, 1]$, $|P(x)| \le e^{|x|/\beta}$ for $x \ge 0$, and $|P(x)| \le e^{|x|}$ for $x \le 0$. Here, we think of $\beta > 1$ as some fixed parameter.

Note that $e^{-x} = 1 - x + \frac{x^2}{2} - \frac{x^3}{6} + \cdots$ satisfies the desired conditions, even for $\eps = 0$, but the issue is that this is not expressible as a polynomial. The natural approach is a Taylor expansion, i.e., we consider the polynomial $P_{\ell}(x) = \sum_{j=0}^{\ell} \frac{(-1)^j x^j}{j!}$ for some integer $\ell$. If $\ell \ge \log(1/\eps)$, this polynomial approximates $e^{-x}$ up to error $\eps$ on $[-1, 1]$. Moreover, it is straightforward to verify that $|P_{\ell}(x)| \le e^{|x|}$ for all $x \in \BR$. The issue is that for positive $x$, we do not always have $|P_{\ell}(x)| \le e^{|x|/\beta}$. For instance, if we set $x = \ell$, then $P_{\ell}(\ell) = \sum_{j=0}^{\ell} \frac{(-1)^j \ell^j}{j!}$. Note that the final term of the sum is $\frac{\ell^{\ell}}{\ell!} \approx e^{\ell}$ in absolute value, and this term will roughly define the growth of the overall sum. Indeed, $P_{\ell}(\ell)$ will grow as roughly $e^{\ell}$ for large values of $\ell$. So, for large $\ell$, we do not even have $P_{\ell}(x) \le e^{0.9 x}$ for all $x \ge 0$, because this claim breaks at roughly $\ell$.

\paragraph{Prior approach.}
The first observation that can be made is that the degree-$\ell$ Taylor approximation $P_{\ell}(x)$, for $x > 0$, is only close to $e^{x}$ for $x = \Theta(\ell)$. The details for why this is will not be relevant now, so we will just briefly explain why this holds. For $x \ge \ell$, we can write $P_{\ell}(x) = \sum_{j=0}^{\ell} \frac{(-1)^j x^j}{j!}$. For $x \gg \ell \ge j$, one can verify that the absolute value of each term, $\frac{x^j}{j!}$, is bounded as $e^{o(x)}$, so overall we will have $|P_{\ell}(x)| \le e^{o(x)}$. Conversely, for $x \le \ell$, we can write $P_{\ell}(x) = e^{-x} - \sum_{j > \ell} \frac{(-1)^j x^j}{j!},$ i.e., we take the full Taylor expansion of $e^{-x}$ and remove the terms of degree beyond $\ell$. This time, for $x \ll \ell \le j$, each of the later terms in absolute value, $\frac{x^j}{j!}$, is $e^{o(x)}$. Finally, the first term $e^{-x}$ is at most $1$.

Overall, it is not too difficult to demonstrate that the degree-$\ell$ Taylor expansion $P_{\ell}(x)$, for $x \ge 0$, has absolute value $e^{\Theta(x)}$ \emph{only if} $x = \Theta(\ell)$. At this point, there is a simple but clever trick to beat the naive bound of $e^{|x|}$. Namely, we can write $e^{-x} = e^{-x/2} \cdot e^{-x/2} \approx P_{\ell_1}(x/2) \cdot P_{\ell_2}(x/2)$, where $\ell_1, \ell_2$ are positive integers such that $\ell_1 \ll \ell_2$. As long as both $\ell_1, \ell_2 \ge \log (1/\eps)$, $P_{\ell_1}(x/2) \cdot P_{\ell_2}(x/2)$ will be a good approximation of $e^{-x}$ on $[-1, 1]$. In addition, each of $|P_{\ell_1}(x/2)|, |P_{\ell_2}(x/2)|$ are uniformly bounded by $e^{|x|/2}$ so the product is bounded by $e^{|x|}$. But, if $\ell_1$ is much smaller than $\ell_2$, then for any $x > 0$, we will either have that $|P_{\ell_1}(x/2)| = e^{o(x)}$ or $|P_{\ell_2}(x/2)| = e^{o(x)}$, from the discussion in the above paragraph. Therefore, for $x > 0$, $|P_{\ell_1}(x/2) \cdot P_{\ell_2}(x/2)| \le e^{|x| \cdot (1/2 + o(1))}$.

The approach of~\cite{bakshi2024quantum} is just a simple generalization of the above trick. Namely, we write $e^{-x} = \underbrace{e^{-x/\beta} \cdot e^{-x/\beta} \cdots e^{-x/\beta}}_{\beta \text{ times}} \approx \prod_{t=1}^{\beta} P_{\ell_t}(x/\beta)$, where $\ell_1, \ell_2, \dots, \ell_{\beta}$ are distinct integers. In order to get an overall upper bound of $e^{x/\beta}$ for all $x > 0$, it will be important for the $\ell_t$'s to be mostly far from each other, i.e., we should not have more than $O(1)$ distinct $\ell_t$'s within an $O(1)$ factor of each other. Hence, we may need each $\ell_t$ to be a constant factor larger than the previous $\ell_{t-1}$. Indeed, the final setting will be roughly $\ell_t = 2^t \cdot \log \frac{1}{\eps},$ which means the overall polynomial $P(x) = \prod_{t=1}^{\beta} P_{\ell_t}(-x/\beta)$ will have degree roughly $2^{\beta} \cdot \log \frac{1}{\eps}$.

\paragraph{Our approach.} While the above approach is nice in that it achieves $O(\log \frac{1}{\eps})$ degree for any fixed $\beta$, the exponential dependence on $\beta$ is somewhat undesirable. Our approach has one similar element, in that we consider a polynomial approximation $P(x/s) \approx e^{-x/s}$ for some large $s$ (that depends polynomially on $\beta$). However, instead of approximating the product of $s$ different copies of $e^{-x/s}$, as done in the prior work, we aim to approximate $e^{-x}$ as $(e^{-x/s})^s \approx (P(x/s))^s$, where we will only use a single polynomial $P$.

The main insight is to note that $x^s$ can be approximated by a polynomial $G(x)$ that has much lower-degree $r \ll s$. This is in fact a well-known result, and is inspired by approximation guarantees provided by Chebyshev polynomials (see \cite[Chapter 3]{approx_theory_survey}). So, we will (roughly) consider $e^{-x} = (e^{-x/s})^s \approx P_{\ell}(x/s)^s \approx G(P_{\ell}(x/s))$. Recalling that the main issue for approximating $P_{\ell}(x)$ as the degree-$\ell$ approximation for $e^{-x}$ happens at $x \approx \ell$, the main issue for us will be at $x \approx s \cdot \ell$, where $P_{\ell}(x/s) \approx e^{\ell}$, so $P_{\ell}(x/s)^s \approx e^{\ell s} \approx e^x$. However, because $G$ has some degree $r \ll s$, the value of $G(P_{\ell}(x/s))$ will actually depend more like $e^{\ell \cdot r} \le e^{\beta x}$, as long as $r \le \beta \cdot s$.

It will turn out that for the polynomial $G$ to be a sufficiently good approximation (so that the approximation is within $\eps$ of $e^{-x}$ in $[-1, 1]$), we will need $\ell$ to be logarithmic in $1/\eps$ and $r \approx \sqrt{s \cdot \log(1/\eps)}$, where $r$ is the degree of $G$. But at the same time, we need $r = \beta \cdot s$: solving gives us $r \approx \beta \log (1/\eps)$ and $s \approx \beta^2 \log(1/\eps)$. The overall degree of $G(P_{\ell}(x/s))$ will thus be $r \cdot \ell \approx \beta \cdot \log^2 \frac{1}{\eps}$.

\paragraph{Reducing the dependence on $\eps$.} While this approach is sufficient to reduce the dependence on $\beta$ significantly, we have increased the dependence on $\eps$ from $\log (1/\eps)$ to $\log^2 (1/\eps)$. While this may not seem significant, the overall algorithm's runtime will be exponential in the degree of the polynomial we construct, and thus a $\log^2 (1/\eps)$ will result in a \emph{quasi-polynomial} dependence on $1/\eps$ for the runtime. So, a $\log (1/\eps)$ dependence is desirable.

Luckily, this fix is actually quite simple. We will just take the polynomial $G(P_{\ell}(x/s))$ and \emph{truncate} the polynomial beyond degree $O(\beta \log(1/\eps))$. In other words, we just remove all higher-degree terms from the expansion of the polynomial around $0$. We prove a strong bound on the coefficients of $G(P_{\ell}(x/s))$, which will be sufficient in showing that the truncated polynomial behaves similarly enough to $G(P_{\ell}(x/s))$ to have both the desired flatness and exponential approximation properties.

\subsubsection{Remainder of the algorithm}

The algorithm will mimic that of~\cite{bakshi2024quantum}. In fact, the improved polynomial we provide can almost directly be plugged into the desired algorithm. However, there is one additional caveat, which is that the polynomial needs to satify a certain ``Sum-of-Squares'' identity. Sum-of-Squares is a powerful algorithmic technique, based on semidefinite programming, that has proven highly effective in many optimization and statistics problems. The idea is that, if one can generate a so-called ``Sum-of-Squares'' proof that a certain estimator is accurate, then one can generate a semidefinite programming relaxation of the estimator, which will be computable in polynomial time, and the estimator remains accurate. In our setting, we must show that a certain 2-variable polynomial $R(x, y)$ (which will be based on the flat polynomial $P$ constructed -- see \Cref{thm:sos-main} for more details) is always nonnegative, and moreover has a ``Sum-of-Squares'' proof of nonnegativity, meaning that $R(x, y)$ can be written as $\sum_{i=1}^M q_i(x, y)^2$ for real polynomials $q_i(x, y)$. Such a result was proven by~\cite{bakshi2024quantum} as well, though their polynomial $R(x, y)$ was based on their higher-degree construction of $P$.

To achieve this Sum-of-Squares proof, we first show that our polynomial $P$ (after some appropriate modification) is always positive and that $P(x) \ge 0.01 \cdot |P'(x)|$. In fact, it is known that the polynomial $P_{\ell} = \sum_{j=0}^\ell \frac{(-1)^j x^j}{j!}$ satisfies this property (for $\ell$ even). Moreover, for any integer $s \ge 1$, $P_{\ell}(x/s)^s$ also satisfies this property. Indeed, $\frac{d}{dx} (P_{\ell}(x/s)^s) = s \cdot P_{\ell}(x/s)^{s-1} \cdot \frac{1}{s} \cdot P_{\ell}'(x/s) = P_\ell(x/s)^{s-1} \cdot P_\ell'(x/s)$, so if $P(x) \ge 0.01 \cdot |P'(x)|$, then 
\[0.01 \cdot \left|\frac{d}{dx} (P_{\ell}(x/s)^s)\right| = 0.01 \cdot P_{\ell}(x/s)^{s-1} \cdot |P_{\ell}'(x/s)| \le P_{\ell}(x/s)^{s-1} \cdot P_{\ell}(x/s) = P_{\ell}(x/s)^s.\]
But what we really need is the same bound for $G(P_{\ell}(x/s))$, where we recall that $G(y)$ is a lower-degree approximation of $y^s$. We will show that the ratio $\left|\frac{G'(y)}{G(y)}\right|$ even smaller than $\frac{s}{y}$, i.e., the ratio $\left|\frac{\frac{d}{dy} (y^s)}{y^s}\right|$. This will help us prove that $G(P_{\ell}(x/s))$ has the desired property.

In reality, we will need a Sum-of-Squares proof for a different polynomial inequality (that some bivariate polynomial $R(x, y) \ge 0$ can be written as a sum of squares). However, inspired by the techniques of~\cite{bakshi2024quantum}, we will in fact prove a black-box conversion from the identity $P(x) \ge 0.01 \cdot |P'(x)|$ into the desired Sum-of-Squares bound. Once this is proven, the rest of the algorithm and analysis is identical to that of~\cite{bakshi2024quantum}, and can be viewed as a black box (see \Cref{thm:bakshi-main}).

\section{Preliminaries}

%
%

\subsection{Chebyshev Polynomials}

For any integer $t \ge 0$, we define $\Phi_t(x)$ to be the degree-$t$ Chebyshev polynomial of the first kind, and define $\Phi_{-t} := \Phi_t$. Recall that $\Phi_0(x) = 1$, $\Phi_1(x) = x$, $\Phi_2(x) = 2x^2-1$, and so on.

We recall some well-known properties of Chebyshev polynomials.

\begin{proposition} \label{prop:chebyshev-recursion}
    For any $t \in \BZ$, $\Phi_{t+1}(x) = 2x \cdot \Phi_t(x) - \Phi_{t-1}(x)$.
\end{proposition}

\begin{proposition} \label{prop:chebyshev-basic}
    If $|x| \le 1$, then $\Phi_t(x) = \cos(t \cdot \arccos x)$. If $|x| > 1$, then if $x = \frac{y+(1/y)}{2},$ then $\Phi_t(x) = \frac{y^t + (1/y^t)}{2}$.

    Importantly, this implies that $|\Phi_t(x)| \le 1$ for all $|x| \le 1$, and $|\Phi_t(x)| \ge 1$ for all $|x| \ge 1$.
\end{proposition}

\begin{proposition} \label{prop:phi-odd-even}
    If $t$ is even, the $\Phi_t$ is an even polynomial. Likewise, if $t$ is odd, then $\Phi_t$ is an odd polynomial.
\end{proposition}

\begin{proposition}[Markov Brothers' Inequality] \label{prop:markov-brothers}
    For any $|x| \le 1$ and any $t \ge 0$, $|\Phi_t'(x)| \le t^2$.
\end{proposition}

\begin{proposition} \label{prop:chebyshev-coeff-bound}
    For $t \ge 0$, all coefficients of the Chebyshev polynomial $\Phi_t$ are at most $(1+\sqrt{2})^t$ in absolute value.
\end{proposition}

\begin{proof}
    The proof is simple by induction. Let the base cases be $t = 0$ and $t = 1$, for which the statement is clearly true. For $t \ge 2$, note that $\Phi_t(x) = 2x \cdot \Phi_{t-1}(x)-\Phi_{t-2}(x)$. By the inductive hypothesis, every coefficient of $2x \cdot \Phi_{t-1}$ is at most $2 \cdot (1+\sqrt{2})^{t-1}$, and every coefficient of $\Phi_{t-2}(x)$ is at most $(1+\sqrt{2})^{t-2}$. So, every coefficient of $\Phi_t(x)$ is at most $2 \cdot (1+\sqrt{2})^{t-1}+(1+\sqrt{2})^{t-2} = (1+\sqrt{2})^t$, as desired.
\end{proof}

Next, for any integer $s \ge 0$, we define $\cD_s$ to be the distribution over $\{-s, -(s-1), \dots, s\}$ where we add together $s$ i.i.d. copies of a uniform $\pm 1$ variable. For any integers $s \ge r \ge 0$, we define
\begin{equation} \label{eq:Q}
    G_{r, s}(x) := \mathop{\BE}_{t \sim \cD_s} \left[\Phi_t(x) \cdot \BI[|t| \le r]\right].
\end{equation}

We note some useful properties about the distribution $\cD_s$ and the polynomials $G_{r, s}$.

\begin{proposition} \cite[Theorem 3.1]{approx_theory_survey} \label{prop:binomial}
    We have that $\BE_{t \sim \cD_s} (\Phi_t(x)) = x^s$.
\end{proposition}

\begin{proposition} \cite[Theorem 3.3]{approx_theory_survey} \label{prop:G-bound-1}
    For all $0 \le r \le s$ and all $|x| \le 1$, we have that $|G_{r, s}(x)-x^s| \le 2e^{-r^2/2s}$.
\end{proposition}

\begin{proposition} \label{prop:G-bound-2}
    Let $s \ge 2$ be even. Then, for all $0 \le r \le s$ and all $|x| \ge 1$, we have $0 \le G_{r, s}(x) \le \min(|x|^s, (2|x|)^r)$.
\end{proposition}

\begin{proof}
    First, note that if $t \sim \cD_s$ and $s$ is even, then $t$ is even with probability $1$. Also, for any even $t$ and $|x| \ge 1$, $\Phi_t(x)$ is positive since $\Phi_t(x) = \frac{y^t+(1/y)^t}{2}$ for some real $y$ and even $t$. Thus,
\[G_{r, s}(x) = \mathop{\BE}_{t \sim \cD_s} \left(\Phi_t(x) \cdot \BI[|t| \le r]\right) \le \mathop{\BE}_{t \sim \cD_s} \left(\Phi_t(x)\right) = x^s,\]
    where the final equality is true by \Cref{prop:binomial}.

    Moreover, for any even $t$, if $x \ge 1$ and $\frac{y+1/y}{2} = x$, then $0 < y, 1/y \le 2x,$ so $0 < \Phi_t(x) = \frac{y^t+(1/y)^t}{2} \le (2x)^t$. Since $\Phi_t$ is an even polynomial by \Cref{prop:phi-odd-even}, for all $|x| \ge 1$, $0 < \Phi_t(x) \le (2|x|)^t$. Thus, for any even $t$, $0 \le \Phi_t(x) \cdot \BI[|t| \le r] \le (2|x|)^r$. Taking the expectation over $t \sim \cD_s$, the claim still holds.
\end{proof}

\subsection{Truncation of Polynomials}

For any integer $\ell \ge 0,$ we define the exponential truncation polynomial as
\begin{equation} \label{eq:P}
    E_\ell(x) := \sum_{j=0}^\ell \frac{(-1)^j \cdot x^j}{j!},
\end{equation}
where $x^0 = 1$ for all $x$ (including $x = 0$).
Also, for any polynomial $A(x) = \sum_{i=0}^n a_i x^i$, we define 
\begin{equation} \label{eq:Trunc}
    \Trunc_k(A)(x) := \sum_{j=0}^{\min(n, k)} a_j x^j.
\end{equation}

We note a series of important facts about the polynomial $E_\ell$.

\begin{proposition} \label{prop:E-bound-1}
    For any even $\ell \ge 0$ and all real $x$, $E_\ell(x) \ge \min(1, e^{-x})$.
\end{proposition}

\begin{proof}
    For $x \le 0$, we can write $E_\ell(x) = 1 + \sum_{j=1}^\ell \frac{(-x)^j}{j!}.$ Every term is nonnegative and the first term is $1$, so $E_\ell(x) \ge 1.$

    For $x \ge 0$, we can prove that $e^{-x} \le E_\ell(x)$ for all even $\ell$, via induction on $\ell$. For $\ell = 0$, $E_\ell(x) = 1 \ge e^{-x}$. A simple integration can verify that $e^{-x} = 1 - x + \int_0^x \int_0^y e^{-z} dz dy$, whereas $E_\ell(x) = 1 - x + \int_0^x \int_0^y E_{\ell-2}(z) dz dy.$ By the inductive hypothesis, $e^{-z} \le E_{\ell-2}(z)$ for all $z \ge 0$, which means that $e^{-x} \le E_\ell(x)$ for all $x \ge 0$.
\end{proof}

\begin{proposition} \label{prop:E-bound-2}
    For any integer $\ell \ge 0$ and any real $x$, $|E_\ell(x)| \le e^{|x|}$.
\end{proposition}

\begin{proof}
    The proof is immediate from the observation that 
\[|E_\ell(x)| \le \sum_{j=0}^\ell \frac{|x|^j}{j!} \le \sum_{j=0}^\infty \frac{|x|^j}{j!} = e^{|x|}. \qedhere\]
\end{proof}

\begin{proposition} \label{prop:E-bound-3}
    For any integer $\ell \ge 2$ and any $x \in [-1, 1]$, $|E_\ell(x) - e^{-x}| \le \frac{|x|^\ell}{\ell!}$.
\end{proposition}

\begin{proof}
    Note that $|E_\ell(x)-e^{-x}| \le \sum_{j=\ell+1}^\infty \left|\frac{x^j}{j!}\right|$. For $|x| \le 1$, this is at most $|x|^{\ell+1} \cdot \sum_{j = \ell+1}^\infty \frac{1}{j!} = |x|^{\ell+1} \cdot \left(\frac{1}{(\ell+1)!} + \frac{1}{(\ell+2)!} + \cdots\right)$, which is at most $\frac{|x|^\ell}{\ell!}.$
\end{proof}

\begin{proposition} \cite[Lemma B.1]{bakshi2024quantum} \label{prop:E-bound-4}
    Let $\ell \ge 2$ be even. Then, for all $x \in \BR$, $|E_{\ell-1}(x)| \le 99 \cdot E_\ell(x)$.
\end{proposition}

\begin{proposition} \label{prop:E-bound-5}
    Let $\ell \ge 2$ be even. Then, $\min_x E_\ell(x) \ge \min(\frac{1}{100}, e^{-\ell})$.
\end{proposition}

\begin{proof}
    If $x \le \ell$, then by \Cref{prop:E-bound-1}, $E_\ell(x) \ge \min(1, e^{-x}) \ge e^{-\ell}.$ Alternatively, if $x > \ell$, then by \Cref{prop:E-bound-4}, $99 E_\ell(x) \ge |E_{\ell-1}(x)|$, which means that $99 \cdot E_\ell(x) + E_{\ell-1}(x) \ge 0$. Moreover, $E_\ell(x) - E_{\ell-1}(x) = \frac{x^\ell}{\ell!} \ge \frac{\ell^\ell}{\ell!} \ge 1$. Adding these two equations together, we have that $E_\ell(x) \ge \frac{1}{100}$. So, for all $x$, $E_\ell(x) \ge \min(\frac{1}{100}, e^{-\ell})$.
\end{proof}

\subsection{Sum-of-Squares}

We recall some basics of the Sum-of-Squares method.

\begin{defn}[Sum-of-Squares polynomial]
    Let $p(x_1, \dots, x_m) \in \BR[x_1, \dots, x_m]$ be a real polynomial over variables $x_1, \dots, x_m$, for some $m \ge 1$. We say that $p(x_1, \dots, x_m)$ is a \emph{Sum-of-Squares (SoS)} polynomial if $p(x_1, \dots, x_m) = \sum_{j=1}^M q_j(x_1, \dots, x_m)^2$, for some positive integer $M$ and polynomials $q_1, \dots, q_M \in \BR[x_1, \dots, x_m]$.
\end{defn}

We now note the definition of \emph{bounded} polynomials, from~\cite{bakshi2024quantum}.

\begin{defn}[Bounded polynomial{~\cite[Definition 2.22]{bakshi2024quantum}}]
\label{def:bound-polynomial}
\label{def:sos-bounded-polynomial}
    A polynomial $p(x_1, \dots, x_m) \in \BR[x_1, \dots , x_m ]$ is \emph{$(d,C)$-bounded} if the following properties hold.
    \begin{enumerate}
        \item $p$ has degree at most $d$.
        \item for each monomial in $p$ of total degree $d'$, its coefficient has magnitude at most $C/(d'!)$. 
    \end{enumerate} 
    A polynomial $p$ is a \emph{$(k,d,C)$-bounded Sum-of-Squares (SoS) polynomial} if $p$ is a sum-of-squares polynomial, $p = q_1^2 + \dots + q_k^2$, and each of the $q_i$'s are $(d,C)$-bounded.
\end{defn}

We note the following basic fact about bounded SoS polynomials.

\begin{proposition}\cite[Claim 2.23]{bakshi2024quantum}\label{claim:basic-composition-properties}
Let $p_1(x_1, x_2)$ be a $(k_1, d_1, C_1)$-bounded SoS polynomial and $p_2(x_1, x_2)$ be a $(k_2, d_2, C_2)$-bounded SoS polynomial. Then,
\begin{enumerate}[label=(\alph*)]
    \item $p_1 + p_2$ is a $(k_1 + k_2,\max(d_1, d_2),\max(C_1, C_2))$-bounded SoS polynomial;
    \item $p_1p_2$ is a $(k_1k_2,d_1 + d_2, (d_1 + d_2+1) \cdot 2^{d_1 + d_2} \cdot C_1C_2)$-bounded SoS polynomial;
    \item For any $t \in [0,1]$, $p_1((1 - t)x_1 + ty_1, (1 - t)x_2+ ty_2) $ is a $(k_1, d_1, C_1)$-bounded SoS polynomial in $x_1,y_1,x_2,y_2$.
\end{enumerate}
\end{proposition}

We also note the following fact about bounded univariate SoS polynomials.

\begin{proposition} \label{prop:SoS-bound-roots}
    Let $p(x)$ be a real, degree-$d$, univariate polynomial. Suppose that $p$ has leading coefficient which is positive and at most $1$. Moreover, suppose that $p$ has all roots bounded in magnitude by some $A$, but $p(x)$ has no real roots. Then, $p(x)$ is a $(2, d/2, (A \cdot d)^{d/2})$-bounded SoS polynomial.
\end{proposition}

\begin{proof}
    First, assume that $p$ is monic, i.e., it has leading coefficient $1$.
    Note that since $p(x)$ has real coefficients but no real roots, we can pair the roots of $p$ into complex conjugates, i.e., $p(x) = \prod_{j=1}^{d/2} (x-z_j) (x-\bar{z}_j)$. So, we can write $p(x) = q(x) \bar{q}(x)$, where $q(x) = \prod_{j=1}^{d/2} (x-z_j)$. By writing $q(x) = q_1(x) + i \cdot q_2(x)$ for real polynomials $q_1, q_2$ of degree at most $d/2$, we have that $p(x) = q_1(x)^2 + q_2(x)^2$.

    Since $p$ is monic and has all roots at most $A$, this means $p$ has degree $j$ coefficient bounded by $A^{d/2 - j} \cdot {d/2 \choose j} \le \frac{(A \cdot d)^{d/2}}{j!}$. Thus, $q_1, q_2$ must have their degree $j$ coefficient bounded by $\frac{(A \cdot d)^{d/2}}{j!}.$ Hence, $q_1$ and $q_2$ are both $(d/2, (A \cdot d)^{d/2})$-bounded, which means that $P = q_1^2+q_2^2$ is a $(2, d/2, (A \cdot d)^{d/2})$-bounded SoS polynomial.

    Finally, if $p(x)$ has leading coefficient $0 < p_d < 1$, we can write $p(x) = p_d \cdot \frac{p(x)}{p_d}$. We have just proven that $\frac{p(x)}{p_d}$ is a $(2, d/2, (A \cdot d)^{d/2})$-bounded SoS polynomial, and because we are scaling by a factor $p_d \in (0, 1)$, $p(x)$ is as well.
\end{proof}

\section{Polynomial Construction}

\subsection{Main Theorem}

Our main goal is to produce a low-degree \emph{flat} approximation to the exponential, similar to the goal in~\cite[Section 4]{bakshi2024quantum}. We aim for a flat polynomial of significantly smaller degree, which will be crucial in reducing the runtime of the final algorithm from doubly exponential in $\beta$ to singly exponential.

First, we recall the definition of a flat exponential approximation.

\begin{defn} \cite[Definition 4.1]{bakshi2024quantum} \label{def:flat}
    Given $\eps \in (0, 1/2)$, $\eta \in (0, 1)$, and $\kappa \ge 1$, we say a polynomial $P(x)$ is a $(\kappa, \eta, \eps)$-\emph{flat exponential approximation} if
\begin{enumerate}
    \item For all $x \in [-\kappa, \kappa]$, $|P(x)-e^{-x}| \le \eps$.
    \item For all $x \in \BR$, $|P(x)| \le \max(1, e^{-x}) \cdot e^{\eta \cdot |x|}$.
\end{enumerate}
\end{defn}

The main theorem we wish to prove is the following.

\begin{theorem} \label{thm:approx-main}
    Let $\eps_0 > 0$ be a sufficiently small constant. For any $\eps \in (0, \eps_0)$ and any $\beta \ge 1$, there exists a polynomial $P(x)$ of degree at most $10^9 \cdot \beta^2 \cdot \log \frac{1}{\eps}$ such that
\begin{enumerate}
    \item $P(x)$ is a $(\beta \log \frac{1}{\eps}, \frac{1}{\beta}, \eps)$-flat exponential approximation.
    \item For all $x \in \BR$, $P(x) > 0$ and $99 \cdot P(x) > |P'(x)|$.
    \item The leading coefficients of both $99 \cdot P + P'$ and $99 \cdot P - P'$ are positive and at most $1$, and all roots of both $99 \cdot P + P'$ and $99 \cdot P - P'$ have magnitude bounded by $e^{10^{14} \cdot \beta^3 \cdot \log^2 (1/\eps)}$.
\end{enumerate}
\end{theorem}

Moreover, it will be straightforward to verify that the polynomial $P$ that we construct is computable in $\poly(\beta, \log \frac{1}{\eps})$ time. (See the discussion after Equation~\eqref{eq:parameter-settings}).

\bigskip

In reality, we will focus on proving a slight modification of the theorem, which will have more convenient guarantees to prove.

\begin{theorem} \label{thm:approx}
    Let $\delta_0 = \eps_0^{100}$ (where $\eps_0$ is the constant from \Cref{thm:approx-main}) be a sufficiently small constant. Let $\beta \ge 1$ and $0 < \delta < \delta_0$ be parameters.
    Then, there exists a polynomial $\hat{P}$ of degree at most $5 \cdot 10^6 \cdot \beta \log \frac{\beta}{\delta}$ such that 
\begin{enumerate}
    \item For all $0 \le x \le 4 \beta \log \frac{1}{\delta}$, $|\hat{P}(x)-e^{-x}| \le \delta$.
    \item For all $x \ge 0$, $|\hat{P}(x)| \le e^{x/(2 \beta)}$.
    \item For all $x \le 0$, $|\hat{P}(x)| \le e^{-x}$.
    \item For all $x \in \BR$, $\hat{P}(x) > 0$ and $99 \hat{P}(x) > |\hat{P}'(x)|$.
    \item The leading coefficients of both $99 \cdot \hat{P} + \hat{P}'$ and $99 \cdot \hat{P} - \hat{P}'$ are positive and at most $\delta$, and all roots of both $99 \cdot \hat{P} + \hat{P}'$ and $99 \cdot \hat{P} - \hat{P}'$ have magnitude bounded by $e^{10^{9} \cdot \beta \cdot \log^2 (\beta/\delta)}$.
\end{enumerate}
\end{theorem}

First, we see why \Cref{thm:approx} implies \Cref{thm:approx-main}.

\begin{proof}[Proof of \Cref{thm:approx-main} from \Cref{thm:approx}]
    For simpicity, we will define $\kappa = \beta \log \frac{1}{\eps}$.
    We will set $\delta = e^{-100 \cdot \kappa}$.
    Note that $\delta \le \eps^{100}$ because $\beta \ge 1$, so for any $\eps < \eps_0$, we automatically have $\delta < \delta_0$.
    Next, for $\hat{P}$ which satisfies \Cref{thm:approx}, define $P(x) := (1 - \delta \cdot e^{\kappa}) \cdot e^{\kappa} \cdot \hat{P}(x + \kappa)$. Note that the degree of $P$ equals the degree of $\hat{P}$, which is at most $5 \cdot 10^6 \cdot \beta \log \frac{\beta}{\delta} \le 10^9 \beta^2 \log \frac{1}{\eps}$.

    Thus, it suffices to prove that if $\hat{P}$ satisfies the four properties of \Cref{thm:approx}, then $P$ satisfies the three properties of \Cref{thm:approx-main}.

    First, note that for any $x \in [-\kappa, \kappa]$, $|\hat{P}(x+\kappa) - e^{-(x+\kappa)}| \le \delta$, by Property 1 of \Cref{thm:approx}. This also implies that $|\hat{P}(x+\kappa)| \le 1+\delta$. Therefore, for any $x \in [-\kappa, \kappa]$,
\begin{align*}
    |P(x)-e^{-x}|
    &= |(1-\delta \cdot e^{\kappa}) \cdot e^{\kappa} \cdot \hat{P}(x+\kappa) - e^{-x}| \\
    &= e^\kappa \cdot |(1-\delta \cdot e^{\kappa}) \cdot \hat{P}(x+\kappa) - e^{-(x+\kappa)}| \\
    &\le e^{\kappa} \cdot \left(|\hat{P}(x+\kappa) - e^{-(x+\kappa)}| + \delta \cdot e^{\kappa} \cdot |\hat{P}(x+\kappa)|\right) \\
    &\le e^{\kappa} \cdot \left(\delta + \delta \cdot e^{\kappa} \cdot (1+\delta)\right) \\
    &\le 3 \delta \cdot e^{2\kappa} \le \eps.
\end{align*}

    Next, note that for all $x \in \BR$, $|P(x)| = (1-\delta \cdot e^{\kappa}) \cdot e^{\kappa} \cdot |\hat{P}(x+\kappa)|$. We can then bound $|P(x)|$ based on four cases.
\begin{itemize}
    \item If $x \le -\kappa$, then $|P(x)| \le e^{\kappa} \cdot |\hat{P}(x+\kappa)| \le e^{-x}$, by Property 3 in \Cref{thm:approx}.
    \item If $-\kappa < x \le 0$, then $0 \le x+\kappa \le \kappa,$ so $|P(x)| \le (1-\delta \cdot e^{\kappa}) \cdot e^{\kappa} \cdot (e^{-(x+\kappa)}+\delta),$ where we used Property 1 in \Cref{thm:approx}. Since $x \le 0$, $e^{-(x+\kappa)}+\delta \le e^{-(x+\kappa)} \cdot (1 + \delta \cdot e^{\kappa})$. So, $|P(x)| \le (1-\delta \cdot e^{\kappa}) \cdot e^{-x} \cdot (1 + \delta \cdot e^{\kappa}) \le e^{-x}$. 
    \item If $0 < x \le 4 \beta \log \frac{1}{\delta} - \kappa$, then $|\hat{P}(x+\kappa)| \le \delta + e^{-(x+\kappa)} \le \delta+e^{-\kappa} = e^{-\kappa} \cdot (1 + e^{\kappa} \cdot \delta)$, where we used Property 1 in \Cref{thm:approx}. So, $|P(x)| \le (1 - \delta \cdot e^{\kappa}) \cdot e^{\kappa} \cdot e^{-\kappa} \cdot (1 + e^{\kappa} \cdot \delta) \le 1,$ which is at most $e^{x/\beta}$.
    \item If $x > 4 \beta \log \frac{1}{\delta} - \kappa$, then $|\hat{P}(x+\kappa)| \le e^{(x+\kappa)/(2\beta)}$, by Property 2 of \Cref{thm:approx}. So, $|P(x)| \le e^{\kappa + (x+\kappa)/(2\beta)}$. However, note that $\kappa = \frac{\log (1/\delta)}{100} \le \frac{x}{200 \beta}$. Thus, $\kappa + \frac{x+\kappa}{2 \beta} \le 2 \kappa + \frac{x}{2 \beta} \le \frac{x}{\beta},$ which means that $|P(x)| \le e^{x/\beta}$.
\end{itemize}

    In summary, $|P(x)-e^{-x}| \le \eps$ for all $x \in [-\kappa, \kappa]$, $|P(x)| \le e^{-x}$ whenever $x \le 0$, and $|P(x)| \le e^{x/\beta}$ whenever $x \ge 0$, so Property 2 in \Cref{def:flat} is satisfied. Therefore, because $\kappa = \beta \log \frac{1}{\eps}$, $P(x)$ is a $(\beta \log \frac{1}{\eps}, \frac{1}{\beta}, \eps)$-flat exponential approximation.

    Next, note that if $\hat{P}(x+\beta) > 0$, then $P(x) > 0$ and $\frac{\hat{P}'(x+\beta)}{\hat{P}(x+\beta)} = \frac{P'(x)}{P(x)},$ since $P(x)$ is just a scaled version of $\hat{P}(x+\beta)$. Thus, if Property 4 of \Cref{thm:approx} holds for $\hat{P}$, then for all $x \in \BR$, $P(x) > 0$ and $|P'(x)| < 99 \cdot P(x)$.

    Finally, note that $(99 \cdot P - P')(x) = (1-\delta \cdot e^{\kappa}) \cdot e^{\kappa} \cdot (99 \hat{P} - \hat{P}')(x+\kappa)$. Thus, the leading coefficient of $99 \cdot P - P'$ is $(1-\delta \cdot e^{\kappa}) \cdot e^{\kappa}$ times the leading coefficient of $99 \cdot \hat{P}-\hat{P}'$, which is positive and at most $e^{\kappa} \cdot \delta \le 1$. Moreover, the roots of $99 \cdot P - P'$ are the same as the roots of $99 \cdot \hat{P}-\hat{P}'$, up to a shift of $\kappa$. So, all roots of $99 \cdot P - P'$, in magnitude, are at most $e^{10^9 \cdot \beta \cdot \log^2 (\beta/\delta)} + \kappa \le e^{10^{9} \cdot \beta \cdot (\log \beta + 100 \beta \log 1/\eps)^2} + \beta \log \frac{1}{\eps} \le e^{10^{14} \cdot \beta^3 \cdot \log^2 (1/\eps)}$. The same calculations can be done for $99 \cdot P + P'$ versus $99 \cdot \hat{P} + \hat{P}'$.
\end{proof}

\subsection{The Polynomial}

For some appropriate choices of $k, \ell, r, s$, we define
\begin{equation} \label{eq:preliminary-polynomial}
    Q(x) := G_{r, s}\left(E_\ell\left(\frac{x}{s}\right)\right)
\end{equation}
and
\begin{equation} \label{eq:main-polynomial}
    \hat{P}(x) := \left(1-\frac{\delta}{5}\right) \cdot \Trunc_k(Q)\left(x\right) + \left(\frac{x}{2 s}\right)^{k+2} + \frac{\delta}{10}.
\end{equation}
One should think of $Q$ as essentially satisfying the desired properties already, but $\hat{P}$ is a necessary modification of $Q$ to further reduce the degree.

We will set
\begin{align}
    \ell &= 2 \left\lceil \log \frac{\beta}{\delta} \right\rceil \nonumber \\
    s &= 2 \left\lceil 10^7 \cdot \beta^2 \cdot \log \frac{\beta}{\delta} \right\rceil \nonumber \\
    r &= 2 \left\lceil 10^4 \cdot \beta \cdot \log \frac{\beta}{\delta} \right\rceil \nonumber \\
    k &= 2 \left\lceil 10^6 \cdot \beta \cdot \log \frac{\beta}{\delta} \right\rceil \label{eq:parameter-settings}
\end{align}
Clearly, $\ell, s, r, k$ are all even. Moreover, note that $\hat{P}$ is straightforward to compute in time $\poly(\ell, s, r, k) = \poly(\beta, \log \frac{1}{\delta})$. Thus, $P(x) = (1-\delta \cdot e^{\kappa}) \cdot e^{\kappa} \cdot \hat{P}(x+\kappa)$, where $\kappa = \beta \log \frac{1}{\eps}$ and $\delta = e^{-100 \kappa}$, can be computed in time $\poly(\beta, \log \frac{1}{\delta}) = \poly(\beta, \log \frac{1}{\eps})$ time.

\subsection{Properties of $Q$}

In this subsection, we will show that $Q(x)$ (Equation \eqref{eq:preliminary-polynomial}) will satisfy some modified versions of the properties in \Cref{thm:approx}. In the next subsection, we show that modifying $Q$ to $\hat{P}$ will precisely satisfy all properties, while having even lower degree.
We will not worry about satisfying Property 5 in \Cref{thm:approx} in this subsection, and will deal with Property 5 in the next subsection.
From now on, we assume the parameter choices in \eqref{eq:parameter-settings}, and assume that $\beta \ge 1$ and $\delta < \delta_0$ for a sufficiently small constant $\delta_0 > 0$.

First, we note a basic fact.

\begin{proposition} \label{prop:tail-event-bound}
    For $r, s$ as in \eqref{eq:parameter-settings}, $e^{-r^2/(2s)} \le (\delta/\beta)^5$.
\end{proposition}

\begin{proof}
    Note that $r \ge 2 \cdot 10^4 \cdot \beta \cdot \log \frac{\beta}{\delta}$ and $s \le 4 \cdot 10^7 \cdot \beta^2 \cdot \log \frac{\beta}{\delta}$. Thus, $\frac{r^2}{2s} \ge 5 \cdot \log \frac{\beta}{\delta}$, so $e^{-r^2/(2s)} \le e^{-5 \log(\beta/\delta)} = (\delta/\beta)^5$.
\end{proof}

\begin{lemma}[Property 1] \label{lem:accuracy}
    For all $0 \le x \le s$, we have that $|Q(x) - e^{-x}| \le \frac{\delta}{50}$.
\end{lemma}

\begin{proof}
    By \Cref{prop:E-bound-3}, we have that $|E_\ell(x/s)-e^{-x/s}| \le \frac{(x/s)^\ell}{\ell!} \le \frac{x/s}{2}$. Therefore,
\begin{equation} \label{eq:E-ell-bound}
    |E_\ell(x/s)| \le e^{-x/s} + \frac{x/s}{2} \le 1,
\end{equation}
    where we use the fact that $e^{-y}+\frac{y}{2} \le 1$ for any $0 \le y \le 1$, and set $y = x/s$.

    Now, we use the standard fact that for any positive integer $s$ and any real values $a, b$, $|a^s-b^s| \le s \cdot |a-b| \cdot \max(|a|^{s-1}, |b|^{s-1})$. Hence, we obtain that $|E_\ell(x/s)^s - e^{-x}| \le s \cdot |E_\ell(x/s) - e^{-x/s}| \cdot \max\left(|E_\ell(x/s)|^{s-1}, |e^{-x/s}|^{s-1}\right)$. Since $|e^{-x/s}| \le 1$ and $|E_\ell(x/s)| \le 1$ by \Cref{eq:E-ell-bound}, this means that $|E_\ell(x/s)^s - e^{-x}| \le s \cdot |E_\ell(x/s) - e^{-x/s}| \le s \cdot \frac{(x/s)^\ell}{\ell!} \le \frac{s}{\ell!}$.
    
    Next, note that $\ell! \ge (\ell/e)^\ell$. So, if $\delta < \delta_0$ is sufficiently small, because $\ell \ge 2 \cdot \log \frac{\beta}{\delta} \ge 2 \cdot \log \delta_0^{-1},$ then $\ell! \ge e^{5 \ell} \ge (\beta/\delta)^5$. Hence, because $s \le 10^8 \cdot \beta^2 \cdot \log \frac{\beta}{\delta} \le 10^8 \cdot \frac{\beta^3}{\delta}$, this means $\frac{s}{\ell!} \le 10^8 \cdot \frac{\delta^4}{\beta^2} \le \frac{\delta}{100}$. In summary, for all $0 \le x \le s$, we have $|E_\ell(x/s)^s-e^{-x}| \le \frac{\delta}{100}$.

    Next, we bound $|E_\ell(x/s)^s - G_{r, s}(E_\ell(x/s))|$. Indeed, because $|E_\ell(x/s)| \le 1$ by Equation~\eqref{eq:E-ell-bound}, we can apply Propositions~\ref{prop:G-bound-1} and~\ref{prop:tail-event-bound} to obtain that $|E_\ell(x/s)^s-G_{r, s}(E_\ell(x/s))| \le 2e^{-r^2/2s} \le 2 \cdot (\delta/\beta)^5 \le \frac{\delta}{100}$.

    So, by triangle inequality, we have that $|G_{r, s}(E_\ell(x/s)) - e^{-x}| \le \frac{\delta}{50}$.
\end{proof}

\begin{lemma}[Property 2] \label{lem:large}
    For any $x \ge s$, $-\frac{\delta}{50} \le G_{r, s}(E_\ell(x/s)) \le e^{0.1 \cdot x/\beta}$.
\end{lemma}

\begin{proof}
    Since $x \ge s \ge 0$, we have $|E_\ell(x/s)| \le e^{x/s}$ by \Cref{prop:E-bound-2}. 

    First, assume that $1 \le |E_\ell(x/s)| \le e^{x/s}$. In this case, by \Cref{prop:G-bound-2}, $0 \le G_{r, s}(E_\ell(x/s)) \le (2 \cdot e^{x/s})^r.$ Since $x \ge s$, $2 \cdot e^{x/s} \le e^{2 x/s}$. Thus, $0 \le G_{r, s}(E_\ell(x/s)) \le e^{(2x/s) \cdot r} \le e^{0.1 \cdot x/\beta}$, by our bounds on $r$ and $s$.
    
    Alternatively, $|E_\ell(x/s)| \le 1$, in which case $|G_{r, s}(E_\ell(x/s)) - E_\ell(x/s)^s| \le 2 \cdot e^{-r^2/2s} \le 2(\delta/\beta)^5$ by Propositions~\ref{prop:G-bound-1} and~\ref{prop:tail-event-bound}. Since $s$ is even, this means $0 \le E_\ell(x/s)^s \le 1$, so $-2 (\delta/\beta)^5 \le G_{r, s}(E_\ell(x/s)) \le 1+2 (\delta/\beta)^5$. Moreover, assuming $\delta < \delta_0$, $2(\delta/\beta)^5 \le \frac{\delta}{50}$ and $2(\delta/\beta)^5 \le \frac{0.1}{\beta} \le \frac{0.1 \cdot x}{\beta}$. So, $-\frac{\delta}{50} \le G_{r, s}(E_\ell(x/s)) \le 1 + \frac{0.1 x}{\beta} \le e^{0.1 \cdot x/\beta}$.

    In either case, we have that $-\frac{\delta}{50} \le G_{r, s}(E_\ell(x/s)) \le e^{0.1 \cdot x/\beta}$.
\end{proof}

\begin{lemma}[Property 3] \label{lem:small}
    For any $x \le 0$, $0 \le G_{r, s}(E_\ell(x/s)) \le e^{|x|}$.
\end{lemma}

\begin{proof}
    By Propositions~\ref{prop:E-bound-1} and~\ref{prop:E-bound-2}, $1 \le E_\ell(x/s) \le e^{|x|/s}$. Therefore, by \Cref{prop:G-bound-2}, $0 \le G_{r, s}(E_\ell(x/s)) \le (e^{|x|/s})^s = e^{|x|}$.
\end{proof}

Before proving that $Q$ satisfies a version of Property 4, we prove several auxiliary lemmas.

\begin{lemma} \label{lem:phi-ratio-1}
    On the range $[1, \infty)$, and for any $t \ge 0$, $\frac{\Phi_{t+1}}{\Phi_t}$ is at least $1$ and increasing.
\end{lemma}

\begin{proof}
    We prove this via induction. For the base case $t = 0$, this equals $\frac{x}{1} = x$, which, on $[1, \infty)$, is clearly increasing and at least $1$.
    Next, for any $t \ge 1$, note that $\Phi_{t+1}(x) = 2x \cdot \Phi_t(x) - \Phi_{t-1}(x),$ which means that $\frac{\Phi_{t+1}(x)}{\Phi_t(x)} = 2x - \frac{\Phi_{t-1}(x)}{\Phi_t(x)}.$ By our inductive hypothesis, $\frac{\Phi_t(x)}{\Phi_{t-1}(x)} \ge 1$, which means $\frac{\Phi_{t-1}(x)}{\Phi_t(x)} \le 1$, so $2x - \frac{\Phi_{t-1}(x)}{\Phi_t(x)} \ge 2 - \frac{\Phi_{t-1}(x)}{\Phi_t(x)} \ge 1$ for all $x \ge 1$. Moreover, by our inductive hypothesis, $\frac{\Phi_t(x)}{\Phi_{t-1}(x)}$ is increasing and positive, which means $\frac{\Phi_{t-1}(x)}{\Phi_t(x)}$ is decreasing, which means $2x - \frac{\Phi_{t-1}(x)}{\Phi_t(x)}$ is increasing. This completes the inductive step.
\end{proof}

\begin{corollary} \label{cor:phi-ratio-2}
    For any $x \ge 1$, we have that $0 = \frac{\Phi_0'(x)}{\Phi_0(x)} \le \frac{\Phi_1'(x)}{\Phi_1(x)} \le \frac{\Phi_2'(x)}{\Phi_2(x)} \le \cdots$.
\end{corollary}

\begin{proof}
    Since $\Phi_0(x) = 1$, clearly $\frac{\Phi_0'(x)}{\Phi_0(x)} = \frac{0}{1} = 0$.

    So, we just need to check that for any $t \ge 0$, $\frac{\Phi_{t+1}'(x)}{\Phi_{t+1}(x)} \ge \frac{\Phi_{t}'(x)}{\Phi_{t}(x)}$. To see why, by \Cref{lem:phi-ratio-1}, $\frac{d}{dx}\left(\frac{\Phi_{t+1}(x)}{\Phi_t(x)}\right) = \frac{\Phi_t(x) \cdot \Phi_{t+1}'(x) - \Phi_t'(x) \cdot \Phi_{t+1}(x)}{\Phi_t(x)^2} \ge 0$. Since $\Phi_t, \Phi_{t+1}$ are strictly positive for $x \ge 1$, this implies that 
\begin{align*}
    0 
    &\le \frac{\Phi_t(x) \cdot \Phi_{t+1}'(x) - \Phi_t'(x) \cdot \Phi_{t+1}(x)}{\Phi_t(x)^2} \cdot \frac{\Phi_t(x)}{\Phi_{t+1}(x)} \\
    &= \frac{\Phi_t(x) \cdot \Phi_{t+1}'(x) - \Phi_t'(x) \cdot \Phi_{t+1}(x)}{\Phi_t(x) \cdot \Phi_{t+1}(x)} \\
    &= \frac{\Phi_{t+1}'(x)}{\Phi_{t+1}(x)} - \frac{\Phi_t'(x)}{\Phi_t(x)}. \qedhere
\end{align*}
\end{proof}

\begin{lemma} \label{lem:G-prime-1}
    For all $x \ge 1$, $0 \le G_{r, s}'(x) \le \frac{s}{x} \cdot G_{r, s}(x)$
\end{lemma}

\begin{proof}
    Fix any $x \ge 1$, and let $a_t := \frac{\Phi_t'(x)}{\Phi_t(x)}$ and $b_t := \Phi_t(x) \cdot \BP_{t' \sim \cD_s}(|t'| = t)$. By \Cref{cor:phi-ratio-2}, $0 = a_0 \le a_1 \le \cdots$ and for every $0 \le t \le s$ and $x \ge 1$, $b_t$ is nonnegative. By \Cref{prop:binomial}, $x^s = \sum_{t=0}^s \BP_{t' \sim \cD_s}(|t'| = t) \cdot \Phi_t(x) = \sum_{t=0}^s b_t$, and $G_{r, s}(x) = \sum_{t'=0}^r \BP_{t' \sim \cD_s}(|t'| = t) \cdot \Phi_t(x) = \sum_{t = 0}^r b_t$. We also have that $s x^{s-1} = \frac{d}{dx} (x^s) = \sum_{t=0}^s \Phi_t'(x) \cdot \BP_{t' \sim \cD_s}(|t'| = t) = \sum_{t=0}^s a_t b_t$, and $G_{r, s}' = \sum_{t=0}^r \Phi_t'(x) \cdot \BP_{t' \sim \cD_s}(|t'| = t) = \sum_{t=0}^r a_t b_t$.

    Since every $b_t$ is nonnegative, and every $a_t$ is nonnegative by \Cref{cor:phi-ratio-2}, this immediately implies that $G'_{r, s}(x) = \sum_{t=0}^r a_t b_t \ge 0$. 
    Next, to prove that $G_{r, s}'(x) \le \frac{s}{x} \cdot G_{r, s}(x)$, it is equivalent to prove that $G_{r, s}'(x) \cdot x^s \le G_{r, s}(x) \cdot sx^{s-1}$, or
\[\left(\sum_{t=0}^r a_t b_t\right) \cdot \left(\sum_{t=0}^s b_t\right) \le \left(\sum_{t=0}^s a_t b_t\right) \cdot \left(\sum_{t=0}^r b_t\right).\]
    The Right Hand Side minus the Left Hand Side of the above equation equals
\[\sum_{t > r, t' \le r} a_t b_t b_{t'} - \sum_{t \le r, t' > r} a_t b_t b_{t'} = \sum_{t > r, t' \le r} (a_t-a_{t'}) b_t b_{t'}.\]
    \Cref{cor:phi-ratio-2} tells us that $a_t-a_{t'} \ge 0$ for all $t > r$ and $t'\le r$, which completes the proof.
\end{proof}

\begin{lemma} \label{lem:G-prime-2}
    For all $0 \le x \le 1$, $|G_{r, s}'(x) - sx^{s-1}| \le 2 s^2 \cdot e^{-r^2/2s}.$
\end{lemma}

\begin{proof}
    By \Cref{prop:binomial}, we can write $G_{r, s}(x) = x^s - \sum_{t > r} \BP_{t' \sim \cD_s} (|t'| = t) \cdot \Phi_t(x)$. Hence, $G_{r, s}'(x) = sx^{s-1} - \sum_{t > r} \BP_{t' \sim \cD_s} (|t'| = t) \cdot \Phi_t'(x)$. By \Cref{prop:markov-brothers}, $|\Phi_t'(x)| \le t^2 \le s^2,$ so $|G_{r, s}'(x) - sx^{s-1}| \le \sum_{t > r} s^2 \cdot \BP_{t' \sim \cD_s} (|t'| = t) \le 2 s^2 \cdot e^{-r^2/2s},$ where the final inequality holds by Hoeffding's inequality.
\end{proof}

We can now prove that $Q$ almost satisfies Property $4$ in \Cref{thm:approx}.

\begin{lemma} \label{lem:derivative-bound}
    We have that for all $x$, $|Q'(x)| \le 99 \cdot Q(x) + \frac{\delta}{50}$.
\end{lemma}

\begin{proof}
    Note that $E_\ell'(x) = -E_{\ell-1}(x)$. Since $\ell$ is even, by \Cref{prop:E-bound-4}, we have that $|\frac{d}{dx} E_\ell(x/s)| = |\frac{1}{s} \cdot E_{\ell-1}(x/s)| \le \frac{99}{s} \cdot |E_\ell(x/s)|$. Then, we have that 
\[|Q'(x)| = \left|\frac{d}{dx} G_{r,s}(E_\ell(x/s))\right| \le |G_{r, s}'(E_\ell(x/s))| \cdot \frac{99}{s} \cdot |E_\ell(x/s)|.\]

    Now, let $y := E_\ell(x/s)$. By \Cref{prop:E-bound-5}, we know that $y \ge \min\left(\frac{1}{100}, e^{-\ell}\right) > 0$.
    
    Next, we bound $|G_{r, s}'(y)| = |G_{r, s}'(E_\ell(x/s))|$. If $y \ge 1$, then by \Cref{lem:G-prime-1}, $G_{r, s}(y)$ is nonnegative and $|G'_{r, s}(y)| \le \frac{s}{y} \cdot G_{r, s}(y)$. Otherwise, $0 < y < 1$, so by \Cref{lem:G-prime-2}, $|G'_{r, s}(y)| \le s y^{s-1} + 2s^2 \cdot e^{-r^2/2s}$, whereas by \Cref{prop:G-bound-1}, $G_{r, s}(y) \ge y^s - 2e^{-r^2/2s}$.

    In the case where $y = E_\ell(x/s) \ge 1$, $|Q'(x)| \le \frac{s}{y} \cdot G_{r, s}(y) \cdot \frac{99}{s} \cdot y = 99 \cdot Q(x).$ Otherwise, 
\[|Q'(x)| \le (s y^{s-1} + 2s^2 \cdot e^{-r^2/2s}) \cdot \frac{99}{s} \cdot y = 99 \cdot y^s + 198 s e^{-r^2/2s} \cdot y \le 99 \cdot G_{r, s}(y) + 198 (s+1) e^{-r^2/2s}.\]
    By~\Cref{prop:tail-event-bound}, $e^{-r^2/2s} = (\delta/\beta)^{5}$, so $198(s+1) e^{-r^2/2s} \le \frac{\delta}{50}$, assuming $\delta < \delta_0$ is sufficiently small. Therefore, whether $y \ge 1$ or $0 < y < 1$, we have that $|Q'(x)| \le 99 \cdot Q(x) + \frac{\delta}{50}$. 
\end{proof}

While the desired properties are not exactly satisfied, it will turn out that a simple shifting/scaling will be enough to modify $Q$ to exactly satisfy the properties of \Cref{thm:approx}.
However, we will not make this correction yet, because we also wish to further reduce the degree of $Q$, which we will do using $\Trunc_k$. Since $G_{r, s}$ has degree $r$ and $E_\ell$ has degree $\ell$, $Q$ has degree $r \cdot \ell = O(\beta \cdot \log^2 \frac{\beta}{\delta})$. While this already improves over \cite{bakshi2024quantum} in that it reduces the dependence on $\beta$ to polynomial rather than exponential, we have picked up an additional factor of $\log \frac{1}{\delta}$, compared to~\cite{bakshi2024quantum}. By truncating, we can further reduce the degree to $k+2 = O(\beta \cdot \log \frac{\beta}{\delta})$ while ensuring that all four properties hold, which will prove \Cref{thm:approx}.

\subsection{Proof of \Cref{thm:approx}}

Before proving \Cref{thm:approx}, we will show that $\hat{P}$, despite having smaller degree $k+2 = O(\beta \cdot \log \frac{1}{\delta})$, is very similar to $Q$.

We start by bounding the coefficients of $Q(x)$.

\begin{lemma} \label{lem:final-truncation}
    For every $j \ge 0$, the degree $j$ coefficient of $Q(x) = G_{r, s}(E_\ell(x/s))$ (as a polynomial in $x$) is at most $5^r \cdot \frac{(r/s)^j}{j!}$ in absolute value. 
\end{lemma}

\begin{proof}
    Note that
\[\left(\sum_{j=0}^\infty \frac{1}{j! \cdot s^j} \cdot x^j\right)^t = \left(\sum_{j=0}^\infty \frac{(x/s)^j}{j!}\right)^t = e^{x \cdot t/s} = \sum_{j=0}^\infty \frac{(t/s)^j}{j!} \cdot x^j.\]
    This holds not only for all $x$, but also as an equality as formal series in $x$.
    Since all of these coefficients $\frac{1}{j! \cdot s^j}$ are positive, if we look at
\[E_\ell(-x/s)^t = \left(\sum_{j=0}^\ell \frac{(x/s)^j}{j!}\right)^t = \left(\sum_{j=0}^\ell \frac{1}{j! \cdot s^j} \cdot x^j\right)^t\]
    as a polynomial, for any $j \ge 0$ the degree $j$ coefficient is nonnegative and at most $\frac{(t/s)^j}{j!}.$ Therefore, the degree $j$ coefficient of $E_\ell(x/s)^t$ is at most $\frac{(t/s)^j}{j!}$ in absolute value.

    Now, let's look at the degree $j$ coefficient of a Chebyshev polynomial of degree $t$ applied to $E_\ell(x/s)$, i.e., $\Phi_t(E_\ell(x/s))$. By \Cref{prop:chebyshev-coeff-bound}, the degree $j$ coefficient, in absolute value, is at most 
\[\sum_{t' = 0}^t (1+\sqrt{2})^t \cdot \frac{(t'/s)^j}{j!} \le (1+\sqrt{2})^t \cdot (t+1) \cdot \frac{(t/s)^j}{j!} \le 5^t \cdot \frac{(t/s)^j}{j!}.\]
    Since $G_{r, s}(x) = \BE_{t \sim \cD_s}[\Phi_t(x) \cdot \BI[|t| \le r]],$ $G_{r, s}(E_\ell(x/s))$ is a weighted average of $\Phi_t(E_\ell(x/s))$ for $0 \le t \le r$ and $0$, which means that the degree $j$ coefficient of $G_{r, s}(E_\ell(x/s))$ is at most $5^r \cdot \frac{(r/s)^j}{j!}$ in absolute value.
\end{proof}

\begin{corollary} \label{cor:final-truncation}
    Let $\Err(x)$ denote the polynomial $Q(x) - \Trunc_k(Q)(x)$. Then, we can bound both
\begin{equation}
    |\Err(x)|, |\Err'(x)| \le 
\begin{cases}
    \frac{\delta}{10^5} & |x| \le 4 s \\
    e^{0.1 |x|/\beta} & |x| > 4 s.
\end{cases}
\end{equation}
\end{corollary}

\begin{proof}
    By \Cref{lem:final-truncation}, we can bound
\begin{equation} \label{eq:err}
    |\Err(x)| \le \sum_{j = k+1}^\infty 5^r \cdot \frac{(r/s)^j}{j!} \cdot |x|^j
\end{equation}
    and
\begin{equation} \label{eq:err-prime}
    |\Err(x)| \le \sum_{j = k+1}^\infty 5^r \cdot \frac{(r/s)^j}{(j-1)!} \cdot |x|^{j-1}.
\end{equation}
    If $|x| \le 4 s$, using the fact that $k \ge 90 r$, \eqref{eq:err} is at most
\[\sum_{j=k+1}^\infty \frac{5^r \cdot (4r)^j}{j!} \le \sum_{j=k+1}^\infty \frac{(10r)^j}{j!} \le \sum_{j=k+1}^\infty \frac{(10r)^j}{(j/e)^j} = \sum_{j=k+1}^\infty \left(\frac{10 e \cdot r}{j}\right)^j \le \sum_{j=k+1}^\infty 2^{-j} = 2^{-k} \le \frac{\delta}{10^5}\]
    and \eqref{eq:err-prime} is at most
\[\sum_{j=k+1}^\infty \frac{(5r/s) \cdot 5^{r-1} \cdot (r \cdot |x|/s)^{j-1}}{(j-1)!} \le \sum_{j=k+1}^\infty \frac{5^{r-1} \cdot (4r)^{j-1}}{(j-1)!} \le \sum_{j=k}^\infty \left(\frac{10 e \cdot r}{j}\right)^j \le \sum_{j=k}^\infty 2^{-j} = 2^{-k+1} \le \frac{\delta}{10^5}.\]

    If $|x| > 4 s$, \eqref{eq:err} is at most
\[5^r \cdot \sum_{j=0}^\infty \frac{(|x| \cdot r/s)^j}{j!} = 5^r \cdot e^{|x| \cdot r/s} \le e^{2|x| \cdot r/s} \le e^{0.1 \cdot |x|/\beta}.\]
    and \eqref{eq:err-prime} is at most
\[5^r \cdot (r/s) \cdot \sum_{j=1}^\infty \frac{(|x| \cdot r/s)^{j-1}}{(j-1)!} \le 5^r \cdot \sum_{j=0}^\infty \frac{(|x| \cdot r/s)^{j}}{j!} \le e^{0.1 \cdot |x|/\beta}. \qedhere\]
\end{proof}

We are now ready to prove the desired properties of \Cref{thm:approx}.

\begin{lemma}[Property 1] \label{lem:property-1}
    For all $0 \le x \le 4 \beta \log \frac{1}{\delta}$, $|\hat{P}(x)-e^{-x}| \le \delta$.
\end{lemma}

\begin{proof}
    We will prove the result for all $0 \le x \le s$: note that $s \ge 4 \beta \log \frac{1}{\delta}$ by \eqref{eq:parameter-settings}.

    We know that $|Q(x)-e^{-x}| \le \frac{\delta}{50}$ for all $0 \le x \le s$, by \Cref{lem:accuracy}. Next, we know that $|Q(x) - \Trunc_k(Q)(x)| \le \frac{\delta}{10^5} \le \frac{\delta}{50}$ for all $0 \le x \le s$, by \Cref{cor:final-truncation}. 
    So, by Triangle inequality, $|\Trunc_k(Q)(x) - e^{-x}| \le \frac{\delta}{25}$.
    This implies that $|\Trunc_k(Q)(x)| \le e^{-x}+\frac{\delta}{25} \le 1 + \frac{\delta}{25},$ since $x$ is nonnegative.
    So, for any $0 \le x \le s$, we have
\begin{align*}
    |\hat{P}(x)-e^{-x}|
    &= \left|\left(1-\frac{\delta}{5}\right) \cdot \Trunc_k(Q)(x) + \left(\frac{x}{2s}\right)^{k+2} + \frac{\delta}{10} - e^{-x} \right| \\
    &\le \left|\Trunc_k(Q)(x) - e^{-x}\right| + \frac{\delta}{5} \cdot |\Trunc_k(Q)(x)| + \left(\frac{|x|}{2s}\right)^{k+2} + \frac{\delta}{10} \\
    &\le \frac{\delta}{25} + \frac{\delta}{5} \cdot \left(1 + \frac{\delta}{25}\right) + 2^{-(k+2)} + \frac{\delta}{10},
\end{align*}
which is at most $\delta$.
\end{proof}

\begin{lemma}[Properties 2 and 3] \label{lem:property-2-3}
    For all $x \ge 0$, $|\hat{P}(x)| \le e^{x/(2 \beta)}$, and for all $x \le 0$, $|\hat{P}(x)| \le e^{-x}$.
\end{lemma}

\begin{proof}
    We will repeatedly use the fact that $|\hat{P}(x)| \le \left(1-\frac{\delta}{5}\right) \cdot |\Trunc_k(Q)(x)| + \left(\frac{|x|}{2s}\right)^{k+2} + \frac{\delta}{10}$, which follows by the definition of $\hat{P}$.

    First, assume $0 \le x \le s$. We saw in the proof of \Cref{lem:property-1} that $|\Trunc_k(Q)(x)| \le 1 + \frac{\delta}{25}$. Thus, 
\[|\hat{P}(x)| \le \left(1 - \frac{\delta}{5}\right) \cdot \left(1 + \frac{\delta}{25}\right) + \left(\frac{|x|}{2s}\right)^{k+2} + \frac{\delta}{10} \le 1 - \frac{\delta}{5} + \frac{\delta}{25} + 2^{-(k+2)}+\frac{\delta}{10} \le 1 \le e^{x/(2\beta)}.\]

    Next, assume $x \ge s$. By \Cref{cor:final-truncation}, we have that $|Q(x) - \Trunc_k(Q)(x)| \le e^{0.1 \cdot x/\beta}$. Moreover, by \Cref{lem:large}, $|Q(x)| \le e^{0.1 \cdot x/\beta}$. So, $|\Trunc_k(Q)(x)| \le 2e^{0.1 \cdot x/\beta}$. Therefore,
\begin{align*}
    |\hat{P}(x)| 
    &\le \left(1 - \frac{\delta}{5}\right) \cdot |\Trunc_k(Q)(x)| + \left(\frac{x}{2s}\right)^{k+2} + \frac{\delta}{10} \\
    &\le 2e^{0.1 \cdot x/\beta} + e^{(x/2s) \cdot (k+2)} + 1 \\
    &\le 4e^{x/(4\beta)} \\
    &\le e^{x/(2\beta)},
\end{align*}
    where the second line uses the fact that $\frac{x}{2s} \le e^{x/2s}$ (indeed, $y \le e^y$ for all $y \in \BR$).

    Next, assume that $-s \le x \le 0$. In this case, $|\Trunc_k(Q)(x)| \le |\Err(x)| + |Q(x)| \le \frac{\delta}{50} + e^{|x|},$ by \Cref{cor:final-truncation} and \Cref{lem:small}. So,
\begin{align*}
    |\hat{P}(x)| 
    &\le \left(1-\frac{\delta}{5}\right) \cdot \left(e^{|x|}+\frac{\delta}{50}\right) + 2^{-(k+2)} + \frac{\delta}{10} \\
    &\le e^{|x|} - \frac{\delta}{5} + \frac{\delta}{50} + 2^{-(k+2)} + \frac{\delta}{10} \\
    &\le e^{|x|}.
\end{align*}

    Finally, assume $x \le -s$. In this case, $|\Trunc_k(Q)(x)| \le |\Err(x)| + |Q(x)| \le e^{0.1 |x|/\beta} + e^{|x|},$ by \Cref{cor:final-truncation} and \Cref{lem:small}. So, 
\begin{align*}
    |\hat{P}(x)| 
    &\le \left(1-\frac{\delta}{5}\right) \cdot \left(e^{|x|}+e^{0.1 \cdot |x|/\beta}\right) + \left(\frac{|x|}{2s}\right)^{k+2} + \frac{\delta}{10} \\
    &\le e^{|x|} - \frac{\delta}{5} \cdot e^{|x|} + e^{0.1 \cdot |x|/\beta} + e^{x \cdot (k+2)/(2s)} + \frac{\delta}{10} \\
    &\le e^{|x|} - \frac{\delta}{5} \cdot e^{|x|} + e^{0.1 \cdot |x|/\beta} + e^{0.5 \cdot |x|} + \frac{\delta}{10} \\
    &\le e^{|x|}. \qedhere
\end{align*}
\end{proof}

\begin{lemma}[Property 4] \label{lem:property-4}
    For all $x$, we have that $99 \cdot \hat{P}(x) > |\hat{P}'(x)|$. (Note that this automatically implies $\hat{P}(x) > 0$.)
\end{lemma}

\begin{proof}
    First, note that for all $x$, $|Q'(x)| \le 99 \cdot Q(x) + \frac{\delta}{50}$, by \Cref{lem:derivative-bound}. 
    
    We start with the case $|x| \le 4s$. By \Cref{cor:final-truncation}, $|\Err(x)|, |\Err'(x)| \le \frac{\delta}{10^5}$, so 
\[99 \cdot \Trunc_k(Q)(x) \ge 99 \cdot Q(x) - \frac{\delta}{10^3} \ge |Q'(x)| - \frac{\delta}{50} - \frac{\delta}{10^3} \ge |\Trunc_k(Q)'(x)| - \frac{\delta}{25}.\]
    Next, we recall that $\hat{P}(x) = (1-\frac{\delta}{5}) \cdot \Trunc_k(Q)(x) + \left(\frac{x}{2s}\right)^{k+2} + \frac{\delta}{10},$ which means that
\begin{align} \label{eq:value}
    99 \cdot \hat{P}(x) &\ge \left(1 - \frac{\delta}{5}\right) \cdot 99 \cdot \Trunc_k(Q)(x) + \frac{\delta}{10} + \left(\frac{x}{2s}\right)^{k+2} \nonumber \\
    &\ge \left(1 - \frac{\delta}{5}\right) \cdot |\Trunc_k(Q)'(x)| + \frac{\delta}{20} + \left(\frac{x}{2s}\right)^{k+2}
\end{align}
    and
\begin{equation} \label{eq:derivative}
    |\hat{P}'(x)| \le \left(1-\frac{\delta}{5}\right) \cdot |\Trunc_k(Q)'(x)| + \frac{k+2}{2s} \cdot \left(\frac{|x|}{2s}\right)^{k+1}.
\end{equation}
    For any $|x| \le k+2$, note that $\left|\frac{k+2}{2s} \cdot \left(\frac{x}{2s}\right)^{k+1}\right|$ is at most $2^{-k} \le \frac{\delta}{50}$, and $k$ is even so $\left(\frac{x}{2s}\right)^{k+2}$ is positive. Thus, \eqref{eq:derivative} is smaller than \eqref{eq:value}. Alternatively, if $k+2 < |x| \le 4s$, then $\left|\frac{k+2}{2s} \cdot \left(\frac{|x|}{2s}\right)^{k+1}\right| \le \left(\frac{|x|}{2s}\right)^{k+2} = \left(\frac{x}{2s}\right)^{k+2},$ so we still have that \eqref{eq:derivative} is smaller than \eqref{eq:value}. In either case, $99 \cdot \hat{P}(x) > |\hat{P}'(x)|$.

    We now assume $|x| \ge 4s$. In this case, we explicitly write
\[\hat{P}(x) = \left(\frac{x}{2s}\right)^{k+2} + \frac{\delta}{10} + \left(1 - \frac{\delta}{5}\right) \cdot \sum_{j=0}^k q_j x^j,\]
    where $q_j$ represent the coefficients of $Q$. (Note that we stop the summation at degree $k$.) By \Cref{lem:final-truncation}, $|q_j| \le 5^r \cdot \frac{(r/s)^j}{j!}$. Hence, assuming $|x| \ge 4s$ and since $k$ is even,
\begin{align*}
    \hat{P}(x) 
    &\ge \left(\frac{|x|}{2s}\right)^{k+2} - \sum_{j=0}^k 5^r \cdot \frac{(r/s)^j}{j!} |x|^j \\
    &= 2^{k+2} \cdot \left(\frac{|x|}{4s}\right)^{k+2} - 5^r \cdot \sum_{j=0}^k \left(\frac{|x|}{4s}\right)^j \cdot \frac{(4r)^j}{j!} \\
    &\ge \left(\frac{|x|}{4s}\right)^{k+2} \cdot \left(2^{k+2} - 5^r \cdot \sum_{j=0}^k \frac{(4r)^j}{j!}\right) \\
    &\ge \left(\frac{|x|}{4s}\right)^{k+2} \cdot \left(2^{k+2} - 5^r \cdot e^{4r}\right) \\
    &\ge \left(\frac{|x|}{4s}\right)^{k+2} \cdot 2^{k+1} = \frac{1}{2} \cdot \left(\frac{|x|}{2s}\right)^{k+2}.
\end{align*}
    Conversely,
\begin{align*}
    |\hat{P}'(x)|
    &\le \frac{k+2}{2s} \cdot \left(\frac{|x|}{2s}\right)^{k+1} + \sum_{j=1}^k 5^r \cdot \frac{(r/s)^j}{(j-1)!} \cdot |x|^{j-1} \\
    &= \frac{k+2}{|x|} \cdot 2^{k+2} \cdot \left(\frac{|x|}{4s}\right)^{k+2} + 5^r \cdot \frac{r}{s} \cdot \sum_{j=0}^{k-1} \left(\frac{|x|}{4s}\right)^j \cdot \frac{(4r)^j}{j!} \\
    &\le \left(\frac{|x|}{4s}\right)^{k+2} \cdot \left(\frac{k+2}{|x|} \cdot 2^{k+2} + 5^r \cdot \frac{r}{s} \cdot \sum_{j=0}^{k-1} \frac{(4r)^j}{j!}\right) \\
    &\le \left(\frac{|x|}{4s}\right)^{k+2} \cdot \left(2^{k+1} + 5^r \cdot e^{4r}\right) \\
    &\le \left(\frac{|x|}{4s}\right)^{k+2} \cdot 2^{k+2} = \left(\frac{|x|}{2s}\right)^{k+2}. 
\end{align*}
    So, we again have that $99 \cdot \hat{P}(x) > |\hat{P}'(x)|$; in fact, we even have that $2 \cdot \hat{P}(x) \ge |\hat{P}'(x)|$.
\end{proof}

Finally, we verify that the roots of $\hat{P}$ are (reasonably) bounded.

\begin{lemma}[Property 5] \label{lem:property-5}
    Then, the leading coefficient of $(99 \cdot \hat{P} - \hat{P}')$ is at most $\delta$ in magnitude, and all roots of $\hat{P}(x)$ have magnitude bounded by $e^{10^9 \cdot \beta \cdot \log^2 (\beta/\delta)}$.
\end{lemma}

\begin{proof}
    Note that for $\hat{P}$, the only term beyond degree $k$ is the term $\left(\frac{x}{2s}\right)^{k+2}$. Thus, the leading coefficient of $(99 \cdot \hat{P} - \hat{P}')$ is $99 \cdot (2s)^{-(k+2)} \le 2^{-(k+2)} \le \delta$.

    Next, by \Cref{lem:final-truncation}, the degree $j$ coefficient of $Q(x)$ is at most $\frac{5^r}{j!}$ in absolute value, which means that every coefficient of $99 \cdot Q(x) - Q'(x)$ is at most $100 \cdot 5^r$ in absolute value.
    Therefore, for every $0 \le j \le k$, the degree $j$ coefficient of $(99 \cdot \hat{P} - \hat{P}')$ is at most $\left(1 - \frac{\delta}{5}\right) \cdot 100 \cdot 5^r + \frac{\delta}{10} \le 100 \cdot 5^r$. Moreover, the degree $k+1$ coefficient is $-\frac{k+2}{(2s)^{k+2}}$ and the degree $k+2$ coefficient is $\frac{99}{(2s)^{k+2}}$.
    
    This implies that any $z$ which is a root of $99 \cdot \hat{P} - \hat{P}'$ is at most $3 \cdot (2s)^{k+2} \cdot 5^r$. This is because if $|z| > 3 \cdot (2s)^{k+2} \cdot 5^r$, then 
\begin{align*}
    |99 \cdot \hat{P}(z) - \hat{P}'(z)| 
    &\ge 99 \cdot (2s)^{-(k+2)} \cdot |z|^{k+2} - (k+2) (2s)^{-(k+2)} \cdot |z|^{k+1} - 100 \cdot 5^r \cdot \sum_{j=0}^k |z|^j \\
    &\ge (2s)^{-(k+2)} \cdot \left(99 \cdot |z|^{k+2} - 100 \cdot 5^r \cdot (2s)^{k+2} \cdot \sum_{j=0}^{k+1} |z|^j\right) \\
    &\ge (2s)^{-(k+2)} \cdot \left(99 \cdot |z|^{k+2} - 200 \cdot 5^r \cdot (2s)^{k+2} \cdot |z|^{k+1}\right) \\
    &> 0,
\end{align*}
    where the final inequality is strict. Thus, $z$ cannot be a root of $\hat{P}$.

    So, the roots are bounded by $3 \cdot (2s)^{k+2} \cdot 5^r \le e^{10^9 \cdot \beta \cdot \log^2 (\beta/\delta)},$ by the parameter settings on $r, s, k$ and the assumption that $\delta < \delta_0$ is sufficiently small.
\end{proof}

Combining everything together, we have \Cref{thm:approx}.

\begin{proof}[Proof of \Cref{thm:approx}]
    Note that $\hat{P}$ has degree $k+2  \le 5 \cdot 10^6 \cdot \beta \cdot \log \frac{\beta}{\delta}$. Moreover, all five properties are satisfied, by Lemmas \ref{lem:property-1}, \ref{lem:property-2-3}, \ref{lem:property-4}, and \ref{lem:property-5}.
\end{proof}

\section{Sum-of-Squares Bound}


In this section, we prove an important Sum-of-Squares result about the polynomial $P$, showing that a particular bivariate polynomial (depending on $P$) can be expressed as a SoS polynomial. The result we prove will correspond to~\cite[Theorem 4.6]{bakshi2024quantum}, and will be a key ingredient in the final algorithm.

Specifically, our goal in this section is to prove the following theorem.

\begin{theorem} \label{thm:sos-main}
    Let $P(x)$ be the polynomial of \Cref{thm:approx-main}. Let $d$ equal the degree of $P(x)$, and and let $U(x)$ be any polynomial satisfying $U'(x) = P(x)$. Then, the polynomial
\[R(x, y) := 0.5 (x-y) (1 + 0.25(x-y)^2) \cdot (U(x)-U(y)) - 0.00025(x-y)^2 P(x)\]
    is a $(6d^2, d, e^{10^{23} \cdot \beta^5 \cdot \log^3(1/\eps)})$-bounded SoS polynomial in $x, y$.
\end{theorem}

Before we prove the above theorem, we note a few key lemmas.

\begin{lemma} \cite[Claim B.4]{bakshi2024quantum} \label{lem:r-integral}
    Let $p(x, y, \lambda)$ be a polynomial such that for all $\lambda \in [0, 1]$, it is a $(k, d, C)$-bounded SoS polynomial in $x, y$ (after plugging in a real value for $\lambda$). Then the polynomial
\[r(x, y) = \int_0^1 p(x, y, \lambda) d\lambda\]
    is a $(3d^2, d, \sqrt{k} C)$-bounded SoS polynomial in $x, y$.
\end{lemma}

\begin{lemma}[Restatement of {\cite[Lemma B.7]{bakshi2024quantum}}]
    Let $p, q$ be univariate polynomials such that $p$ is the derivative of $q$. Let $p'$ be the derivative of $p$. Define
\begin{equation} \label{eq:r}
    r(x, y) := 0.5 (x-y) (1+0.25(x-y)^2)(q(x)-q(y))- 0.00025 (x-y)^2 (p(x)+p(y)).
\end{equation}
    Define the polynomials $z = (x+y)/2$ and $a = (x-y)/2$. Then,
\begin{multline} \label{eq:r-integral}
    r(x, y) = \int_0^1 \left(\left(0.998 a^2+a^4\right) p\left(z + \lambda a\right) - 0.001 a^3 (2-\lambda) \cdot p'\left(z + \lambda a\right)\right) d \lambda \\
    + \int_0^1 \left(\left(0.998 a^2+a^4\right) p\left(z - \lambda a\right) - 0.001 a^3 (2-\lambda) \cdot p'\left(z - \lambda a\right)\right) d \lambda.
\end{multline}
\end{lemma}

\begin{lemma} \label{lem:ACBD}
    Let $A, B, C, D$ be polynomials such that $A+B, A-B, C+D, C-D$ are all $(k, d, C)$-bounded SoS polynomials. Then, $A \cdot C+B \cdot D$ and $A \cdot C-B \cdot D$ are both $(2k^2, 2d, 2^{3d} C^2)$-bounded SoS polynomials.
\end{lemma}

\begin{proof}
    Define $Q_1 := A-B$ and $Q_2 := B-(-A) = A+B$. Define $R_1 := C-D$ and $R_2 := D-(-C) = C+D$. Note that $Q_1, Q_2, R_1, R_2$ are all $(k, d, C)$-bounded SoS polynomials. Moreover, we can write
\[A = \frac{Q_1+Q_2}{2},\, B = \frac{Q_2-Q_1}{2},\, C = \frac{R_1+R_2}{2},\, D = \frac{R_2-R_1}{2}.\]
    As a result, we can easily compute
\[A \cdot C - B \cdot D = \frac{Q_1 \cdot R_2 + Q_2 \cdot R_1}{2},\, A \cdot C + B \cdot D = \frac{Q_1 \cdot R_1 + Q_2 \cdot R_2}{2}.\]
    
    By Part b) of \Cref{claim:basic-composition-properties}, all of $Q_1 \cdot R_1, Q_2 \cdot R_2, Q_1 \cdot R_2, Q_2 \cdot R_1$ are $(k^2, 2d, (2d+1) \cdot 2^{2d} C^2)$-bounded SoS polynomials, Thus, $A \cdot C - B \cdot D$ and $A \cdot C + B \cdot D$ are $(2k^2, 2d, \frac{2d+1}{2} \cdot 2^{2d} C^2)$-bounded SoS polynomials. This also means they are $(2k^2, 2d, 2^{3d} C^2)$-bounded SoS polynomials, since $\frac{2d+1}{2} \le 2^d$ for all $d \ge 1$.
\end{proof}

We can now prove the following claim, which roughly states that as long as $p(x)$ satisfies $|p'(x)| \le 99 \cdot p(x)$ for all $x \in \BR$, then in fact $r(x, y)$ is not only always nonnegative but has a sum-of-squares proof of nonnegativity. While in the proof we assume $99 \cdot p(x) + p'(x)$ and $99 \cdot p(x) - p'(x)$ have Sum-of-Squares proofs of nonnegativity, we remark that a nonnegative univariate polynomial is always expressible as a sum of squares. 

\begin{lemma} \label{lem:1d-sos}
    Suppose that $p, q$ are univariate polynomials such that $p = q'$. Define $r(x, y)$ as in \eqref{eq:r}. Suppose that $99 \cdot p(x) + p'(x)$ and $99 \cdot p(x) - p'(x)$ are both $(k, d/2, C)$-bounded SoS polynomials, where $k \ge 2$, $d \ge 4,$ and $C \ge 10$. Then, $r(x, y)$ is a $(6 d^2, d, k \cdot 2^{2d} \cdot C^2)$-bounded SoS polynomial.
\end{lemma}

\begin{proof}
    We start by considering the polynomials
\begin{equation} \label{eq:SoS-1}
    \left(0.998 a^2+a^4\right) p\left(z + \lambda a\right) - 0.001 a^3 (2-\lambda) \cdot p'\left(z + \lambda a\right)
\end{equation}
    and 
\begin{equation} \label{eq:SoS-2}
    \left(0.998 a^2+a^4\right) p\left(z - \lambda a\right) - 0.001 a^3 (2-\lambda) \cdot p'\left(z - \lambda a\right)
\end{equation}
    for any fixed $0 \le \lambda \le 1$, where $z := \frac{x+y}{2}$ and $a := \frac{x-y}{2}$. The bounds on these polynomials are identical, so we will focus on the first of them.

    Note that $p(z+\lambda a) = p\left(\frac{x+y}{2} + \lambda \cdot \frac{x-y}{2}\right) = p\left(\frac{1+\lambda}{2} \cdot x + \frac{1-\lambda}{2} \cdot y\right)$. Likewise, $p'(z+\lambda a) = p'\left(\frac{1+\lambda}{2} \cdot x + \frac{1-\lambda}{2} \cdot y\right)$. Thus, by part c) of \Cref{claim:basic-composition-properties}, $99 \cdot p(z+\lambda a) - p'(z+\lambda a)$ and $99 \cdot (z + \lambda a) + p'(z + \lambda a)$, viewed as polynomials in $x$ and $y$, are $(k, d/2, C)$-bounded SoS polynomials.

    Next, note that for $0 \le \lambda \le 1$, 
\begin{align*}
    0.998 a^2 + a^4 - 0.099(2-\lambda)a^3 
    &= a^2 \cdot \left(a - 0.0495 (2-\lambda)\right)^2 + a^2 \cdot \left(0.998 - 0.0495^2 (2-\lambda)^2\right) \\
    &= \left(a \cdot (a - 0.0495(2-\lambda))\right)^2 + \left(a \cdot \sqrt{0.998 - 0.0495^2 (2-\lambda)^2}\right)^2
\end{align*}
    and 
\begin{align*}
    0.998 a^2 + a^4 + 0.099(2-\lambda)a^3 
    &= a^2 \cdot \left(a + 0.0495 (2-\lambda)\right)^2 + a^2 \cdot \left(0.998 - 0.0495^2 (2-\lambda)^2\right) \\
    &= \left(a \cdot (a + 0.0495(2-\lambda))\right)^2 + \left(a \cdot \sqrt{0.998 - 0.0495^2 (2-\lambda)^2}\right)^2 .
\end{align*}
    For $0 \le \lambda \le 1$, note that $0 < 0.998 - 0.0495^2 (2-\lambda)^2 < 1$. Therefore, for $a = \frac{x-y}{2}$ and $0 \le \lambda \le 1$, $a \cdot (a - 0.0495(2-\lambda))$ is $(2, 10)$-bounded, and $a \cdot \sqrt{0.998 - 0.0495^2 (2-\lambda)^2}$ is $(1, 1)$-bounded. Thus, both $0.998 a^2 + a^4 - 0.099(2-\lambda)a^3$ and $0.998 a^2 + a^4 + 0.099(2-\lambda)a^3$ are $(2, 2, 10)$-bounded SoS polynomials in $x, y$.
    
    Therefore, if we write $A = p(z+\lambda a), B = \frac{1}{99} \cdot p'(z+\lambda a),$ $C = 0.998 a^2 + a^4$, and $D = 0.099 (2-\lambda) a^3,$ we have that $A+B, A-B, C+D, C-D$ are all $(k, d/2, C)$-bounded SoS polynomials. So, by \Cref{lem:ACBD}, we can rewrite Equation \eqref{eq:SoS-1} as
\[A \cdot C - B \cdot D = (0.998 a^2 + a^4) p(z+\lambda a) - 0.001 (2-\lambda) a^3 \cdot p'(z+\lambda a),\]
    which is a $(2k^2, d, 2^{3d/2} C^2)$-bounded SoS polynomial, for any $0 \le \lambda \le 1$. A nearly identical calculation will also tell us that Equation \eqref{eq:SoS-2} is also a $(2k^2, d, 2^{3d/2} C^2)$-bounded SoS polynomial, for any $0 \le \lambda \le 1$.

    Therefore, by \eqref{eq:r-integral} and \Cref{lem:r-integral}, we have that $r(x, y)$ is the sum of two terms, each of which is a $(3 d^2, d, \sqrt{2k^2} \cdot 2^{3d/2} C^2)$-bounded SoS polynomial. So, overall, $r(x, y)$ is a $(6 d^2, d, k \cdot 2^{2d} \cdot C^2)$-bounded SoS polynomial.
\end{proof}

By combining \Cref{lem:1d-sos} with what we know about $P(x)$, we can prove \Cref{thm:sos-main}.

\begin{proof}[Proof of \Cref{thm:sos-main}]
    We will apply \Cref{lem:1d-sos}, plugging in $P$ for $p$ and $U$ for $q$. By Property 2 of \Cref{thm:approx-main}, both $99 \cdot P(x) - P'(x)$ and $99 \cdot P(x) + P'(x)$ are positive for all $x \in \BR$, and thus neither polynomial has any real roots. In addition, by Property 3 of \Cref{thm:approx-main}, along with \Cref{prop:SoS-bound-roots}, $99 \cdot P + P'$ and $99 \cdot P - P'$ are both $\left(2, d/2, (e^{10^{14} \cdot \beta^3 \cdot \log^2(1/\eps)} \cdot d)^{d/2}\right)$-bounded SoS polynomials. Then, we can apply \Cref{lem:1d-sos}, to say that the polynomial
\[0.5 (x-y) (1 + 0.25(x-y)^2) \cdot (U(x)-U(y)) - 0.00025(x-y)^2 P(x)\]
    is $\left(6d^2, d, k \cdot 2^{2d} \cdot (e^{10^{14} \cdot \beta^3 \cdot \log^2(1/\eps)} \cdot d)^d\right)$-bounded. Since 
\[d = k+2 \le 3 \cdot 10^6 \cdot \beta \log \frac{\beta}{\delta} = 3 \cdot 10^6 \cdot \beta (\log \beta + 100 \beta \log \frac{1}{\eps}) \le 6 \cdot 10^8 \cdot \beta^2 \log \frac{1}{\eps},\]
    we have that
\[k \cdot 2^{2d} \cdot (e^{10^{14} \cdot \beta^3 \cdot \log^2(1/\eps)} \cdot d)^d \le e^{10^{23} \cdot \beta^5 \cdot \log^3(1/\eps)}. \qedhere\]
\end{proof}


\section{Putting Everything Together}

Finally, we can prove \Cref{thm:main}. First, we note the following result from~\cite{bakshi2024quantum}.

\begin{theorem}[Implicit from~\cite{bakshi2024quantum}] \label{thm:bakshi-main}
    Let $C > 1 > c$ be some absolute constants.
    Let $\beta > 1$, $\eps \in (0, 1),$ and $\bit, d \ge 1$ be parameters which may depend on $\beta, \eps$.
    Suppose $P(x)$ is a polynomial of degree $d$ with the following guarantees.
\begin{enumerate}
    \item $P$ is a $(C \beta \log(1/\eps), c/\beta, c \eps)$-flat exponential approximation.
    \item Let $U$ be the polynomial with $U(x) = \int_0^x P(t) dt$. Then,
\begin{equation} \label{eq:sos-equation-goal}
    0.5(x-y)(1+0.25(x-y)^2)(U(x)-U(y)) - 0.00025 (x-y)^2 P(x)
\end{equation}
    is a $(\bit, d, e^{\bit})$-bounded sum-of-squares polynomial in $x, y$.
\end{enumerate}
    Then, there exists an algorithm for quantum Hamiltonian learning up to error $\eps$ (with $2/3$ probability) which only requires
\[\samples = O\left(m^6 \cdot e^{O(d)} + \frac{\log m}{\beta^2 \eps^2}\right)\]
    copies of the Gibbs state and runtime
\[O\left((m \cdot \bit)^{O(1)} \cdot e^{O(d)} + \frac{m \log m}{\beta^2 \eps^2}\right).\]
\end{theorem}

The proof of \Cref{thm:main} is now quite straightforward, by combining Theorems~\ref{thm:approx-main}, \ref{thm:sos-main}, and \ref{thm:bakshi-main}.

\begin{proof}[Proof of \Cref{thm:main}]
    We assume that $\beta \ge \beta_c$ for some critical threshold $\beta_c$. Otherwise, we can use the known bounds of~\cite{haah2022optimal}, which uses $\samples = \frac{\log m}{\beta^2 \eps^2}$ copies of the Gibbs state, with runtime $\frac{m \log m}{\beta^2 \eps^2}$.
    
    First, set $\beta' = \beta \cdot \max(C, 1/c, 1/\beta_c)$, where $C, c$ are the constants in \Cref{thm:bakshi-main}, and set $\eps' = \min(c, \eps_0) \cdot \eps$, where $\eps_0$ is the constant in \Cref{thm:approx-main}. Then, $\beta' \ge 1$ and $\eps \le \eps_0$. Now, in $\poly(\beta', \log \frac{1}{\eps'}) = \poly(\beta, \log \frac{1}{\eps})$ time (since $C, c, \eps_0$ are all constants), we can construct a polynomial $P$ satisfying \Cref{thm:approx-main}, with respect to $\beta', \eps'$. Importantly, $P$ is a $(C \beta \log(1/\eps), c/\beta, c \eps)$-flat exponential approximation, since $\beta' \ge C \beta,$ $1/\beta' \le c/\beta$, and $\eps' \le c \eps$. Moreover, $P$ has degree $d = O(\beta^2 \cdot \log \frac{1}{\eps})$, since $C, c, \eps_0$ are absolute constants.
    
    Then, by \Cref{thm:sos-main}, we have that for $U(x) = \int_0^x P(t) dt$, where we view $U(x)$ as a degree-$(d+1)$ polynomial, we have that Equation \eqref{eq:sos-equation-goal} is a $(6d^2, d, e^{O(\beta^5 \log^3(1/\eps)})$-bounded SoS polynomial in $x, y$.
    Therefore, if we set $\bit$ to be a sufficiently large multiple of $\beta^5 \cdot \log^3 (1/\eps)$, we have that Equation \eqref{eq:sos-equation-goal} is a $(\bit, d, e^{\bit})$-bounded SoS polynomial in $x, y$.
    
    In summary, the conditions of \Cref{thm:bakshi-main} are met, so the number of samples needed is 
\[\samples = O\left(m^6 \cdot e^{O(d)} + \frac{\log m}{\beta^2 \eps^2}\right) = O\left(m^6 \cdot (1/\eps)^{O(\beta^2)}\right),\]
    and the runtime (along with constructing $P$) is
\[\poly\left(\beta, \log \frac{1}{\eps}\right) + O\left((m \cdot \beta^5 \cdot \log^3(1/\eps))^{O(1)} \cdot e^{O(d)} + \frac{m \log m}{\beta^2 \eps^2}\right) = O\left(m^{O(1)} \cdot (1/\eps)^{O(\beta^2)}\right),\]
    where we assumed that $\beta \ge \beta_c$.
\end{proof}

\section*{Acknowledgements}

The author thanks Ainesh Bakshi, Allen Liu, and Ankur Moitra for several helpful conversations about their work~\cite{bakshi2024quantum}.

\bibliographystyle{alpha}

\begin{thebibliography}{WGFC14}

\bibitem[AAKS20]{anshu2020sample}
Anurag Anshu, Srinivasan Arunachalam, Tomotaka Kuwahara, and Mehdi Soleimanifar.
\newblock Sample-efficient learning of quantum many-body systems.
\newblock In {\em IEEE Annual Symposium on Foundations of Computer Science}. IEEE Computer Society, 2020.

\bibitem[AAKS21]{anshu2021efficient}
Anurag Anshu, Srinivasan Arunachalam, Tomotaka Kuwahara, and Mehdi Soleimanifar.
\newblock Efficient learning of commuting hamiltonians on lattices.
\newblock {\em Electronic notes}, 25, 2021.

\bibitem[BAL19]{bairey2019learning}
Eyal Bairey, Itai Arad, and Netanel~H Lindner.
\newblock Learning a local hamiltonian from local measurements.
\newblock {\em Physical review letters}, 122(2):020504:1--5, 2019.

\bibitem[BLMT24]{bakshi2024quantum}
Ainesh Bakshi, Allen Liu, Ankur Moitra, and Ewin Tang.
\newblock Learning quantum hamiltonians at any temperature in polynomial time.
\newblock In {\em Symposium on Theory of Computing (STOC)}, 2024.

\bibitem[EHF19]{evans2019scalable}
Tim~J Evans, Robin Harper, and Steven~T Flammia.
\newblock Scalable bayesian hamiltonian learning.
\newblock {\em arXiv preprint arXiv:1912.07636}, 2019.

\bibitem[Has10]{hastings2010locality}
Matthew~B Hastings.
\newblock Locality in quantum systems.
\newblock {\em Quantum Theory from Small to Large Scales}, 95:171--212, 2010.

\bibitem[HKT22]{haah2022optimal}
Jeongwan Haah, Robin Kothari, and Ewin Tang.
\newblock Optimal learning of quantum hamiltonians from high-temperature gibbs states.
\newblock In {\em 2022 IEEE 63rd Annual Symposium on Foundations of Computer Science (FOCS)}, pages 135--146. IEEE, 2022.

\bibitem[LR72]{liebrobinson}
Elliott~H. Lieb and Derek~W. Robinson.
\newblock The finite group velocity of quantum spin systems.
\newblock {\em Communications in Mathematical Physics}, 28(3):251–--257, 1972.

\bibitem[SV14]{approx_theory_survey}
Sushant Sachdeva and Nisheeth~K. Vishnoi.
\newblock Faster algorithms via approximation theory.
\newblock {\em Foundations and Trends{\textregistered} in Theoretical Computer Science}, 9(2):125--210, 2014.

\bibitem[WGFC14]{wiebe2014hamiltonian}
Nathan Wiebe, Christopher Granade, Christopher Ferrie, and David~G Cory.
\newblock Hamiltonian learning and certification using quantum resources.
\newblock {\em Physical review letters}, 112(19):190501:1--5, 2014.

\bibitem[WPS{\etalchar{+}}17]{wang2017experimental}
Jianwei Wang, Stefano Paesani, Raffaele Santagati, Sebastian Knauer, Antonio~A Gentile, Nathan Wiebe, Maurangelo Petruzzella, Jeremy~L O’brien, John~G Rarity, Anthony Laing, et~al.
\newblock Experimental quantum hamiltonian learning.
\newblock {\em Nature Physics}, 13(6):551--555, 2017.

\end{thebibliography}

\newcommand{\etalchar}[1]{$^{#1}$}

\end{document}